\documentclass{SciPost}

\usepackage{bbm}
\usepackage{tikz}
\usetikzlibrary{arrows.meta,bending}
\usetikzlibrary{calc}
\usepackage{graphicx,subcaption}
\usepackage{mathtools}
\usepackage{bm}
\usepackage{enumitem}
\usepackage{verbatim}
\usepackage{amssymb}
\usepackage{booktabs}
\usepackage{slashed}
\usepackage{float}
\usepackage[scr = dutchcal]{mathalpha}
\usepackage{lipsum}
\usepackage{multirow}
\usepackage{hyperref}
\usepackage{yfonts}
\usepackage[normalem]{ulem}
\usepackage[dvipsnames]{xcolor}
\usepackage{amsthm}
\usepackage{orcidlink}

\newcommand{\bfmtb}{\bfmtb}
\renewcommand{\bfmtb}{\lambda}

\renewcommand{\Re}{\operatorname{Re}}
\renewcommand{\Im}{\operatorname{Im}}

\DeclareMathOperator{\Tr}{Tr}

\DeclareMathAlphabet{\mathbfsf}{OT1}{cmss}{bx}{n}

\newtheorem{proposition}{Proposition}
\newtheorem{lemma}{Lemma}

\newtheorem{theorem}{Theorem}
\newtheorem{definition}{Definition}

\newcommand{\CH}{\mathcal{H}}

\newcommand{\CW}{\mathcal{W}}

\newcommand{\SCE}{\mathscr{E}}

\newcommand{\N}{\mathbb{N}}

\newcommand{\C}{\mathbb{C}}
\newcommand{\R}{\mathbb{R}}

\newcommand{\del}{\partial}

\renewcommand{\Im}{\text{Im}}
\renewcommand{\Re}{\text{Re}}

\newcommand\be{\begin{equation}}
\newcommand\ee{\end{equation}}
\newcommand\beq{\begin{equation}}
\newcommand\eeq{\end{equation}}
\newcommand\bea{\begin{eqnarray}}
\newcommand\eea{\end{eqnarray}}

\renewcommand{\a}{\alpha}
\renewcommand{\b}{\beta}
\newcommand{\s}{\sigma}
\renewcommand{\t}{\tau}

\newcommand{\ket}[1]{|#1\rangle}

\newcommand{\I}{\mathbf{I}}
\newcommand{\J}{\mathbf{J}}

\newcommand{\sfQ}[1]{\mathsf{Q}\left(#1\right)}

\newcommand{\dkdvI}{\boldsymbol{\mathcal{I} }}

\newcommand{\tbar}{\bar{\tau}}

\definecolor{KCLRed}{RGB}{225,39,38}

\hypersetup{
    colorlinks,
    linkcolor={red!50!black},
    citecolor={blue!50!black},
    urlcolor={blue!80!black}
}

\DeclareSymbolFont{usualmathcal}{OMS}{cmsy}{m}{n}
\SetSymbolFont{usualmathcal}{bold}{OMS}{cmsy}{b}{n}
\DeclareSymbolFontAlphabet{\mathcal}{usualmathcal}
\fancypagestyle{SPstyle}{
\fancyhf{}
\lhead{\colorbox{KCLRed}{\bf \color{white} ~KCL }}

\fancyfoot[C]{\textbf{\thepage}}
}

\begin{document}

\pagestyle{SPstyle}

\begin{center}{\Large \textbf{\color{BrickRed}{
Modular Properties of Generalised Gibbs Ensembles:\\The Deformed KdV Hierarchy\\
}}}\end{center}

\begin{center}\textbf{
Faisal Karimi\,\orcidlink{0009-0003-5591-9911}
\textsuperscript{1$\star$}
}\end{center}

\begin{center}
{\bf 1} Department of Mathematics, King's College London,\\
Strand, London, WC2R 2LS, United Kingdom
\\[\baselineskip]
$\star$\href{mailto:faisal.karimi@kcl.ac.uk}{\small faisal.karimi@kcl.ac.uk}

\end{center}

\section*{\color{BrickRed}{Abstract}}
\textbf{\boldmath{%
We construct a set of conserved charges which are deformations of quantum KdV charges. They are associated with the modular transformation of a GGE with the spin-$\s$ KdV charge inserted with chemical potential $\alpha_\s$. The charges are formal power series in the chemical potential, with coefficients given by local integrals of motion. We define local recursion relations and equivalent flow equations, prove commutativity to all orders, and establish uniqueness of the currents modulo total derivatives for a fixed generating series. We show that the eigenvalues of these deformed charges can be obtained from the thermodynamic Bethe Ansatz in two models. In the $c=1/2$ free fermion theory, the deformed charges reduce to explicit linear combinations of undeformed KdV charges, and their ground-state eigenvalues agree with the TBA expansion to all orders. In the Lee–Yang model, we find numerical agreement through cubic order in the spin-$5$ deformation for the first few charges evaluated on low-lying states.
}}

\vspace{\baselineskip}








\vspace{10pt}
\noindent\rule{\textwidth}{1pt}
\tableofcontents
\noindent\rule{\textwidth}{1pt}
\vspace{10pt}

\section{Introduction}
\subsection*{Ideas and motivation}

The quantum KdV hierarchy in a two-dimensional conformal field
theory (CFT) is an infinite family of mutually commuting local
charges constructed from the stress tensor~\cite{Bazhanov:1994ft}.
Coupling these charges to chemical potentials extends the ordinary
Gibbs ensemble to a generalised Gibbs ensemble (GGE). Suppressing anti-chiral contributions, this is
\begin{equation}
    \Tr_\CH(q^{L_0-\frac{c}{24}}e^{\sum_{n \geq 1}\alpha_{2n-1} \I_{2n-1}})
\end{equation}
where $\CH$ is the CFT state space and $\I_{2n-1}$ are the KdV charges. Unlike the ordinary torus partition function, in which all the chemical potentials $\a_{2n-1}$ vanish, a GGE with a KdV charge inserted has a non-trivial modular transformation law: after exchanging the two
cycles of the torus, the result is not, in general, another GGE built only from the original KdV charges. Furthermore, in the perturbative regime, it was found that the transformation can be written as a GGE with infinitely many non-commuting local charges inserted \cite{Downing:2024nfb}. This can be thought of as a deformation of the CFT Hamiltonian. This raises the question of what becomes of the full integrable hierarchy in the transformed channel. 

The thermodynamic Bethe Ansatz (TBA) provides a useful way to
approach this question. A TBA description of the modular
transformed GGE was developed in~\cite{Downing:2024nfb}. In the
undeformed theory, moments of the TBA solution reproduce the
eigenvalues of the KdV charges. The deformed TBA also produces an
infinite sequence of moments, suggesting that it may encode more
than a deformed Hamiltonian. The main question of this paper is
whether that Hamiltonian belongs to an entire hierarchy of mutually commuting local charges, and whether the higher TBA moments compute their eigenvalues.

To formulate this question, it is important to distinguish the full modular transform from its perturbative expansion around vanishing chemical potential, which is asymptotic. In the $c=1/2$ free fermion theory, the full transform contains non-perturbative Bethe roots whose energies cannot be identified with those of states in the original
CFT~\cite{Downing:2021mfw,Downing:2023lop}. Similar behaviour occurs in the Lee-Yang model~\cite{Downing:2024nfb}. A physical
description uses a translation-invariant defect~\cite{Downing:2023lnp,Karimi:2026tvh}. Exchanging the torus cycles converts a defect insertion in the original trace into a defect present in the spatial quantisation. The transformed trace is taken over a defect Hilbert space, rather than the original CFT Hilbert space.

In this work we focus on the asymptotic regime. In the examples above, the additional non-perturbative contributions to the GGE trace are invisible to the perturbative expansion in chemical potential near zero. The remaining perturbative contributions can be expressed as a trace over the original CFT Hilbert space, with a
deformed Hamiltonian. An all-orders local expression at generic central charge for this
Hamiltonian was conjectured in~\cite{Downing:2025huv} and proved
in~\cite{Ashok:2026cyi}. This provides the starting point for our
construction. Our aim is to construct not just a Hamiltonian, but a full commuting hierarchy, in the perturbative regime.

We answer the question of whether this deformed hierarchy of charges exists in the positive. For each undeformed
KdV charge, we construct a formal series whose coefficients
are cylinder zero-modes of local fields. The individual coefficients in general do not, and do not need to, commute with one another. The full deformed charges
commute to all orders at a common value of the deformation
parameter (the chemical potential), and their first member is the Hamiltonian obtained from
the asymptotic modular transformation. We give both local recursion relations for the deformed currents and an equivalent formulation
in terms of flow equations.

The identification of the higher TBA moments with the
eigenvalues of this hierarchy is tested in two examples. In the
$c=1/2$ free fermion theory, each deformed charge can be expressed
to all orders as a formal linear combination of undeformed KdV
charges. Its ground-state eigenvalues agree coefficient by
coefficient with the corresponding TBA moments. In the interacting
Lee-Yang model, we consider a spin-$5$ deformation and compare
the spin-$1$, spin-$5$, and spin-$7$ charges with the deformed TBA
through cubic order, on the ground state and three low-lying
excited states. These comparisons provide evidence for the spectral
interpretation of the deformed TBA, complementing the all-orders
proof of commutativity.

Related deformations of integrable hierarchies arise in
other settings. KdV charges under $T\bar T$ deformations were
studied in~\cite{LeFloch:2019wlf}, where flow equations for
their eigenvalues were obtained. Generalised $T\bar T$ deformations
and higher-spin charge eigenvalues have also been analysed through
a mirror-channel formulation of GGEs and their TBA
equations~\cite{Hernandez-Chifflet:2019sua}. These results provide a point of comparison with the present construction. Here the deformation is fixed by the Hamiltonian arising from the asymptotic modular transformation of a chiral GGE, rather than by a full current-current deformation. Furthermore, the equations in this work are concerned with the operator algebra of the CFT directly. Whether the similarities between these charge flows have a conceptual relation remains an open question.

\subsection*{Setup}

We now fix our conventions and recall the deformed Hamiltonian
that the hierarchy constructed below extends.

\medskip
\noindent\emph{KdV currents and cylinder charges.}
Holomorphic fields define conserved currents and hence conserved
charges. On a cylinder of circumference $2\pi$, the quantum KdV
charges are of the form
\begin{equation}
    \I_{2n-1}
    =\int_0^{2\pi i}J_{2n}(w)\frac{dw}{2\pi i}.
\end{equation}
We use the radial quantisation map $z=e^w$, with $w$ the cylinder
coordinate and $z$ the plane coordinate. The current $J_{2n}$ is
chosen to be quasi-primary and has conformal weight $2n$, while
the subscript $2n-1$ denotes the spin of the corresponding charge. We use the BLZ normalisation, in which the currents of the KdV hierarchy are monic polynomials in the stress tensor,
\begin{equation}
    J_{2n}=(T^n)+\cdots,
\end{equation}
that is, we have coefficient one for the leading stress-tensor term. Here
$(T^n)$ denotes the right-nested normal-ordered product of $n$
stress tensors, $(T(T...(TT))$. The charges are fixed by demanding they mutually commute,
\begin{equation}
    [\I_s,\I_r]=0,
    \qquad r,s=1,3,5,7,\ldots.
\end{equation}
Writing $(AB)_p$ for the coefficient of the $p$th-order pole in the
operator product expansion (OPE) of $A$ and $B$, the corresponding
local statement is
\begin{equation}
    (J_{2n}J_{2m})_1=0\pmod{\del}.
\end{equation}
The first-order pole is a total derivative, whose cylinder
zero-mode vanishes. Products in round brackets denote conformal normal ordering, and $\del$ denotes differentiation with respect to its argument (i.e. $\del \phi(w) = \del_w \phi (w)$). Suppressing the coordinate dependence, the first few currents are\footnote{These expressions are often given without the total derivative terms required to make them quasi-primary, since such terms do not affect the integrated charges. We include them here.}
\begin{equation}
\begin{aligned}
    J_2&=T,\\
    J_4&=(TT)-\frac3{10}\del^2T,\\
    J_6&=(T(TT))-\frac{c+2}{12}(\del T\,\del T)\\
    &\quad+\del\!\left(
    \frac{2c-23}{54}(\del T\,T)
    +\frac{5c-98}{4536}\del^3T\right).
\end{aligned}
\end{equation}

\medskip
\noindent\emph{GGEs and modular transformation.}
The ordinary torus partition function is
\begin{equation}
    Z(\t,\tbar)
    =\Tr_{\CH}\!\left(
    q^{L_0-\frac{c}{24}}
    \bar q^{\bar L_0-\frac{c}{24}}
    \right),
    \qquad q=e^{2\pi i\t},
\end{equation}
where $\CH$ is the CFT Hilbert space. Modular invariance requires,
in particular,
\begin{equation}
    Z\!\left(-\frac1\t,-\frac1\tbar\right)
    =Z(\t,\tbar).
\end{equation}
Suppressing the antiholomorphic contribution, the KdV GGE is
\begin{equation}\label{eq: GGE trace}
    \Tr_{\CH}\!\left(
    q^{L_0-\frac{c}{24}}
    e^{\sum_{n=1}^{\infty}\alpha_{2n-1}\I_{2n-1}}
    \right),
\end{equation}
with a chemical potential $\alpha_{2n-1}$ coupled to each charge. If we have infinitely many charges, the trace expression converges if $\Re(\alpha_{2n-1})\leq 0$ for all $\alpha_{2n-1}$. If the expression is truncated to a finite number of charges in the GGE, the expression converges as long as the highest spin charge, say the spin $2N-1$ charge, has a chemical potential which satisfies $\Re(\alpha_{2N-1}) \leq 0$.

The defect description of~\eqref{eq: GGE trace} makes the change
of channel explicit. Define
\begin{equation}
    \mathbf D
    =\exp\!\left(
        \sum_{n\geq1}\alpha_{2n-1}\I_{2n-1}
    \right).
\end{equation}
For a purely imaginary modular parameter,
Figure~\ref{fig: Defect picture} represents the torus as a rectangle
of length $L$ and width $R$, with opposite edges identified. In
channel~(I), it is quantised as a cylinder of circumference $R$ and
length $L$, with its ends identified and $\mathbf D$ inserted along
a spatial slice. Exchanging the two cycles gives channel~(II),
where the spatial slice intersects the defect.  The defect considered here is not
topological and does not preserve Virasoro symmetry. Quantisation in
this channel produces a defect Hilbert space $\CH_D$ and a
Hamiltonian $\mathbf H_D$. Following the notation of the figure, channel~(I) has modular
parameter $\hat\t=iL/R$, while channel~(II) has $\t=iR/L$, so that
$\hat\t=-1/\t$. With $\hat q=e^{2\pi i\hat\t}$, the two descriptions
are related by
\begin{equation}
    \Tr_{\CH}\!\left(
    \hat q^{L_0-\frac{c}{24}}\mathbf D
    \right)
     = 
    \Tr_{\CH_D}\!\left(q^{\mathbf H_D}\right).
\end{equation}

\begin{figure}[t]
\centering
\begingroup
\definecolor{cylOutline}{HTML}{303B48}
\definecolor{cylSurface}{HTML}{F6F8FA}
\definecolor{cylRed}{HTML}{C43D3D}
\definecolor{cylGreen}{HTML}{278067}
\definecolor{cylBlue}{HTML}{3267AA}

\begin{tikzpicture}[
  x=1.12cm, y=1cm,
  line cap=round,
  line join=round,
  outline/.style={
    draw=cylOutline,
    line width=.65pt
  },
  seam/.style={
    draw=cylOutline!45,
    line width=.45pt,
    dash pattern=on 1.6pt off 1.8pt
  },
  insertion/.style={
    draw=cylRed,
    line width=1pt
  },
  hidden insertion/.style={
    insertion,
    dash pattern=on 1.7pt off 1.8pt
  },
  evolution/.style={
    line width=.95pt,
    -{Stealth[length=2.2mm,width=1.6mm]}
  },
  dimension/.style={
    draw=cylOutline!85,
    line width=.45pt,
    {Stealth[length=1.6mm,width=1.1mm]}-{Stealth[length=1.6mm,width=1.1mm]}
  },
  witness/.style={
    draw=cylOutline!40,
    line width=.35pt
  },
  every node/.style={inner sep=2pt}
]

\node[font=\large] at (0,7.62) {I};
\node[font=\large] at (6.8,7.62) {II};

\path[fill=cylSurface]
  (-.8,4.3) -- (-.8,7)
  arc[start angle=180,end angle=360,x radius=.8,y radius=.22]
  -- (.8,4.3)
  arc[start angle=0,end angle=-180,x radius=.8,y radius=.22]
  -- cycle;

\draw[seam] (.8,4.3)
  arc[start angle=0,end angle=180,x radius=.8,y radius=.22];

\draw[outline] (-.8,7) -- (-.8,4.3)
  arc[start angle=180,end angle=360,x radius=.8,y radius=.22]
  -- (.8,7);

\draw[outline,fill=white]
  (0,7) ellipse[x radius=.8,y radius=.22];

\draw[evolution,cylGreen] (0,4.09) -- (0,6.78);

\draw[hidden insertion] (.8,5.13)
  arc[start angle=0,end angle=180,x radius=.8,y radius=.22];

\draw[insertion] (-.8,5.13)
  arc[start angle=180,end angle=360,x radius=.8,y radius=.22];

\draw[
  draw=cylOutline,
  line width=.8pt,
  -{Stealth[length=2.7mm,width=2mm]}
] (1.65,5.65) -- (4.45,5.65);

\path[fill=cylSurface]
  (5.45,4.85) -- (8.15,4.85)
  arc[start angle=270,end angle=90,x radius=.22,y radius=.8]
  -- (5.45,6.45)
  arc[start angle=90,end angle=270,x radius=.22,y radius=.8]
  -- cycle;

\draw[seam] (5.45,4.85)
  arc[start angle=270,end angle=450,x radius=.22,y radius=.8];

\draw[outline] (8.15,4.85) -- (5.45,4.85)
  arc[start angle=270,end angle=90,x radius=.22,y radius=.8]
  -- (8.15,6.45);

\draw[outline,fill=white]
  (8.15,5.65) ellipse[x radius=.22,y radius=.8];

\draw[hidden insertion] (6,4.85)
  arc[start angle=270,end angle=450,x radius=.22,y radius=.8];

\draw[insertion] (6,6.45)
  arc[start angle=90,end angle=270,x radius=.22,y radius=.8];

\draw[
  draw=cylBlue,
  line width=.95pt,
  -{Stealth[length=2.2mm,width=1.6mm,bend]}
] (6.8,4.85)
  arc[start angle=270,end angle=110,x radius=.26,y radius=.8];

\draw[outline] (-.8,.5) rectangle (.8,3.2);
\draw[insertion] (-.8,1.33) -- (.8,1.33);

\draw[outline] (5.45,1.05) rectangle (8.15,2.65);
\draw[insertion] (6,1.05) -- (6,2.65);

\draw[witness] (-.89,.5) -- (-1.25,.5);
\draw[witness] (-.89,3.2) -- (-1.25,3.2);
\draw[dimension] (-1.18,.5)
  -- node[left=3pt] {$L$} (-1.18,3.2);

\draw[witness] (-.8,.41) -- (-.8,.04);
\draw[witness] (.8,.41) -- (.8,.04);
\draw[dimension] (-.8,.12)
  -- node[below=3pt] {$R$} (.8,.12);

\draw[witness] (5.36,1.05) -- (5,1.05);
\draw[witness] (5.36,2.65) -- (5,2.65);
\draw[dimension] (5.07,1.05)
  -- node[left=3pt] {$R$} (5.07,2.65);

\draw[witness] (5.45,.96) -- (5.45,.59);
\draw[witness] (8.15,.96) -- (8.15,.59);
\draw[dimension] (5.45,.67)
  -- node[below=3pt] {$L$} (8.15,.67);

\node[text=cylGreen] at (0,-.83)
  {$\displaystyle \hat{\tau}=\frac{iL}{R}=-\frac{1}{\tau}$};

\node at (0,-1.8)
  {$\displaystyle \Tr_{\CH}(\mathbf D\,\hat q^{L_0-\frac{c}{24}}) $};

\node[text=cylBlue] at (6.8,-.83)
  {$\displaystyle \tau=\frac{iR}{L}$};
\node at (6.8,-1.8)
  {$\displaystyle \Tr_{\CH_D}( q^{\mathbf{H}_D}) $};
\end{tikzpicture}
\caption{Interpretation of the modular transformed GGE traces: on torus (I), the GGE is given by a defect inserted as an operator $\mathbf D$ in the trace; on torus (II) the defect is rotated and the transformed GGE is given by a trace over the Hilbert space $\CH_D$ with a defect Hamiltonian $\mathbf{H}_D$ inserted in the trace. The arrows indicate the direction of time.}
\label{fig: Defect picture}
\endgroup
\end{figure}
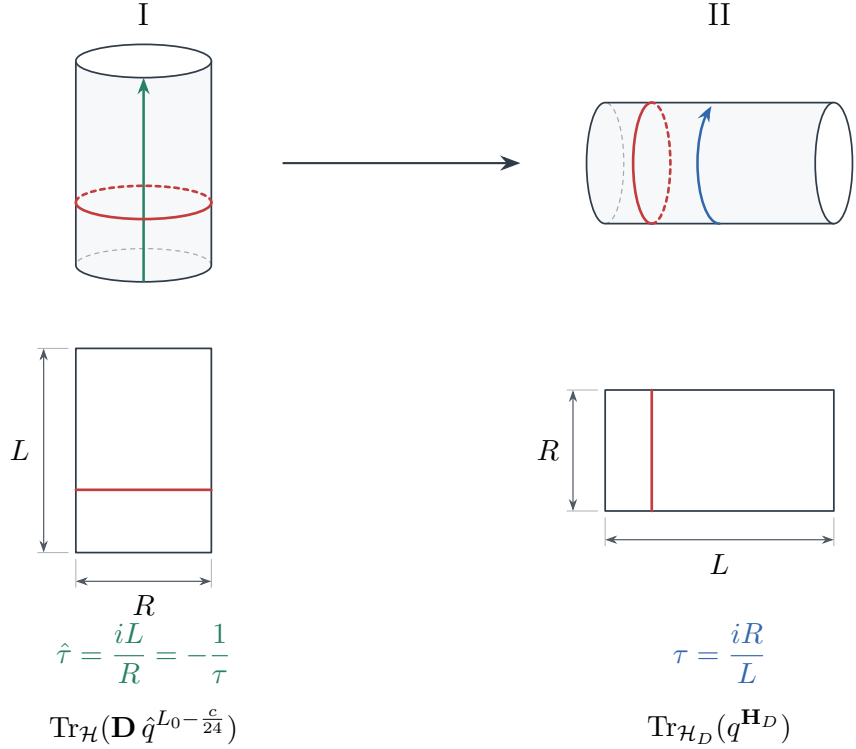

\medskip
\noindent\emph{The deformed Hamiltonian.}
For a single insertion $\I_\s$, the perturbative contribution to
the modular transformation of the GGE can be represented by a trace over the original CFT Hilbert space:
\begin{equation}
    \Tr_{\CH}\!\left(
    q^{L_0-\frac{c}{24}}e^{\alpha_\s\I_\s}
    \right)
    \xrightarrow{S}
    \Tr_{\CH_D}\!\left(q^{\mathbf H_D}\right)
    \underset{\alpha_\s\to0}{\sim}
    \Tr_{\CH}\!\left(q^{\dkdvI_1^{(\s)}}\right).
\end{equation}
We call $\dkdvI_1^{(\s)}$ the \emph{$\s$-deformed Hamiltonian}.
It is obtained by expanding the GGE in $\alpha_\s$, performing
the modular transformation term by term, retaining at each order
the appropriate leading contribution in the modular parameter
corresponding to a one-point thermal function of a local charge, and re-exponentiating the resulting series. The resummation of the full modular transformed series into an exponential was discussed and tested to first non-trivial order in the chemical potential in Sections~3 and~4 of~\cite{Downing:2024nfb}. The all-orders construction of the deformed Hamiltonian found in~\cite{Downing:2025huv,Ashok:2026cyi} says that it is to be written as an infinite sum of local charges
\begin{equation}\label{eq: deformed Hamiltonian}
    \dkdvI_1^{(\s)}
    =\I_1+
    \sum_{n=1}^{\infty}
    \frac{\gamma^n}{(n-1)!\,2^{n-1}}\J_{1,1 + n(\s-1)},
    \qquad
    \gamma=\frac{\alpha_\s\t^\s}{2\pi i}.
\end{equation}
Where each $\J_{1,1+n(\s-1)}$ has spin $1+n(\s-1)$ and is the integral of a quasi-primary field on the cylinder $Y_{1,2+n(\s-1)}(w)$. In general it was found that
\begin{equation}
    [\J_{1,s },\J_{1,r}]\neq0.
\end{equation}
We treat~\eqref{eq: deformed Hamiltonian} as a formal power series. Our interest is the asymptotic regime $\alpha_\s\to0$, and we must not assume convergence of the expansion. Indeed, since the GGE expression converges only on one side of the complex $\alpha_\s$ plane, then any perturbative expansion around $\alpha_\s\to 0$ must be asymptotic since such an expansion will have zero radius of convergence. Furthermore, the corresponding
series is seen to be divergent explicitly in the free theories where it can be determined exactly~\cite{Downing:2021mfw, Karimi:2026tvh}.

\subsection*{Outline}
The paper is organised as follows. Section~\ref{sec: The Deformed Hierarchy} defines the $\s$-deformed KdV hierarchy and gives the local recursion relations for its currents. Section~\ref{sec: The Deformed Hierarchy Commutes} proves that the resulting charges commute to all orders in the deformation parameter, first directly from the recursion and then using the equivalent flow equations. Section~\ref{sec: The Deformed Hierarchy in FF} evaluates the construction in the $c=1/2$ free fermion model, where the operator series can be obtained explicitly and its ground-state eigenvalues agree with the TBA expansion to all orders. Section~\ref{sec: The Deformed Hierarchy in LY} studies the Lee-Yang model, where we find numerical agreement with the proposed TBA eigenvalues through cubic order for the spin-$1$, spin-$5$, and spin-$7$ charges on the ground state and the first three levels of the $h=-1/5$ module. The appendices collect the local fields, their eigenvalues, and the TBA equations used in these comparisons.

\section{The $\sigma$-Deformed KdV Hierarchy}\label{sec: The Deformed Hierarchy}

In this section, we define the $\s$-deformed KdV hierarchy.
Each deformed charge is a formal power series in $\gamma\propto \alpha_\s$,
with coefficients given by integrals of local fields. These
coefficient charges need not commute with one another.
Nevertheless, for any $s$ and $r$, the coefficient of every
power of $\gamma$ in the commutator \([\dkdvI_s^{(\s)},\dkdvI_r^{(\s)}]\) vanishes. Therefore, the full set forms a commuting hierarchy as we prove in
Section~\ref{sec: The Deformed Hierarchy Commutes}.

We begin with a generalised Gibbs ensemble (GGE) containing a single spin-$\s$ KdV charge $\I_\s$, with chemical potential $\alpha_\s$. In the asymptotic regime $\alpha_\s\to0$, its modular $S$-transformation, $\t\mapsto-1/\t$, takes the form
\begin{equation}
    \Tr_{\CH}\!\left(
    q^{L_0-\frac{c}{24}}e^{\alpha_\s\I_\s}
    \right)
    \xrightarrow{S}  \Tr_{\CH_D}\!\left(q^{\mathbf H_D}\right)
     \underset{\alpha_\sigma \to 0}{\sim}
    \Tr_{\CH}\!\left(q^{\dkdvI_1^{(\s)}}\right).
\end{equation}
We call $\dkdvI_1^{(\s)}$ the $\s$-deformed Hamiltonian. The quasi-primary KdV current whose cylinder zero-mode is $\I_\s$ will be called the \emph{generating field}:
\begin{equation}
    W=J_{\s+1}.
\end{equation}
An expression for the deformed Hamiltonian in terms of $W$ was conjectured in \cite{Downing:2025huv}\footnote{Actually, in that paper they were looking only at the case of the zero mode of the spin-$3$ primary field of the $\CW_3$ algebra inserted into the GGE, corresponding to the first non-trivial charge of the quantum Boussinesq hierarchy \cite{Bazhanov:2001xm}. The formula for the modular transformed Hamiltonian doesn't depend on this however and is easily generalised as shown in the follow-up paper mentioned in the main body of the text.} and proved in a follow-up paper~\cite{Ashok:2026cyi}. Although the construction in those references does not require $W$ to be a KdV current, we impose this restriction here because our aim is to deform the entire commuting KdV hierarchy in such a way as to preserve commutativity.

The convention for radial quantisation used in this paper is $z=e^w$, where $w$ is the cylinder coordinate and $z$ is the plane coordinate. For a field $A$, define its cylinder zero-mode by
\begin{equation}
    \sfQ{A}:=\int_0^{2\pi i}A(w)\frac{dw}{2\pi i}.
\end{equation}
This is denoted by $(\hat A)_0$ in~\cite{Downing:2025huv,Karimi:2026tvh}. Since total derivatives integrate to zero on the cylinder, fields that differ by a total derivative define the same charge.

Let $(AB)_p$ denote the coefficient of the $p$th-order pole in the operator product expansion (OPE) of two chiral fields,
\begin{equation}
    A(w_1)B(w_2)\sim
    \sum_{p>0}\frac{(AB)_p(w_2)}{(w_1-w_2)^p}.
\end{equation}
Here we're only interested in the singular part of the OPE. The second-order pole is a commutative product modulo total derivatives, although not generally associative. We define the bracketed fields $[W^n]$ recursively by
\begin{equation}\label{eq: W recursion}
    [W]=W,
    \qquad
    [W^n]=\sum_{k=1}^{n-1}\binom{n-2}{k-1}
    \bigl([W^k][W^{n-k}]\bigr)_2,
    \qquad n\geq2.
\end{equation}
In terms of these fields, the deformed Hamiltonian is $\dkdvI_1^{(\s)}=\sfQ{\mathsf J_1}$, where
\begin{equation}\label{eq: Definition Hamiltonian deformed}
    \mathsf J_1:=J_2+
    \sum_{n=1}^{\infty}\frac{\gamma^n}{(n-1)!\,2^{n-1}}[W^n]
    ,
    \qquad
    \gamma=\frac{\alpha_\s\t^\s}{2\pi i},
\end{equation}
and $J_2 = T$, the stress tensor. We now extend this construction to the higher KdV charges. For each undeformed charge $\I_s$, we take from it the associated quasi-primary current $J_{s+1}$ to define a \emph{seed field}. The deformed currents will be built by acting with the generating field recursively on the seed fields. Define the first contribution to the deformation as
\begin{equation}\label{eq: Ms1}
    M_{s,1}
    =\frac{s}{s+\s}(J_{s+1}W)_2
    =\del^{-1}(J_{s+1}W)_1
    \pmod{\del}.
\end{equation}
If one chooses $s=1$, then
\begin{equation}
    M_{1,1} = W,
\end{equation}
like above. For $n\geq2$, define the higher-order contributions recursively by
\begin{equation}\label{eq: Msn recursion}
    M_{s,n}
    =\frac{1}{n}\del^{-1}\Bigg\{
    \bigl(J_{s+1}[W^n]\bigr)_1
    +2(n-1)\sum_{k=1}^{n-1}\binom{n-2}{k-1}
    \bigl(M_{s,k}[W^{n-k}]\bigr)_1
    \Bigg\}
    \pmod{\del}.
\end{equation}
Here $\del^{-1}$ denotes a local primitive of the complete expression to which it is applied. We use fixed homogeneous representatives of the bracketed fields when evaluating the recursion. The existence of these local primitives is established in Section~\ref{sec: The Deformed Hierarchy Commutes}.

\begin{definition}
For a fixed generating field $W=J_{\s+1}$ and deformation parameter $\gamma=\frac{\alpha_\s\t^\s}{2\pi i}$, the \emph{$\s$-deformed KdV charges} are the formal series
\begin{equation}\label{eq: dkdv definition}
    \dkdvI_s^{(\s)}:=\I_s+
    \sum_{n=1}^{\infty}\frac{\gamma^n}{(n-1)!\,2^{n-1}}
    \sfQ{M_{s,n}}.
\end{equation}
Equivalently, they are the cylinder zero-modes of the deformed currents
\begin{equation}
    \mathsf J_s:=J_{s+1}+
    \sum_{n=1}^{\infty}\frac{\gamma^n}{(n-1)!\,2^{n-1}}M_{s,n}
    \pmod{\del},
    \qquad
    \dkdvI_s^{(\s)}=\sfQ{\mathsf J_s}.
\end{equation}
We refer to $\dkdvI_s^{(\s)}$ as the $\s$-deformed KdV charge of spin-$s$.
\end{definition}
It is convenient for use later in this work to collect the bracketed fields into the generating series
\begin{equation}
    \mathsf K:=\sum_{m=0}^{\infty}
    \frac{\gamma^m}{2^m m!}[W^{m+1}].
\end{equation}
Proposition~\ref{prop: relation M fields} below shows that $M_{1,n}=[W^n]$ modulo total derivatives. Thus the $s=1$ member of the hierarchy reproduces the deformed Hamiltonian in~\eqref{eq: Definition Hamiltonian deformed}, with
\begin{equation}
    \mathsf J_1=T+\gamma\mathsf K
    \pmod{\del},
    \qquad T=J_2.
\end{equation}
 
Two qualifications are necessary. First, we treat the series
in~\eqref{eq: dkdv definition} as formal and do not assume
convergence. Indeed, the Hamiltonian expansion is divergent
in the free fermion and symplectic fermion examples studied
in~\cite{Downing:2021mfw,Karimi:2026tvh}. In these free theories,
the TBA equations are trivial, and determining the excited-state
energies reduces to finding the roots of a polynomial equation.
The branch connected to the undeformed spectrum reproduces
the perturbative expansion of the deformed Hamiltonian.
The energies associated with the remaining branches diverge
as $\alpha_\s\to0$, while their exponentiated contributions to
the GGE have vanishing asymptotic expansions. Therefore these contributions do not appear at any order in the power series in
$\alpha_\s$. In later sections, we compare the eigenvalues of
the higher deformed charges with the corresponding expansions
obtained from the TBA perturbatively. A non-perturbative completion of these
charges lies beyond the scope of the present work.

Second, the label ``spin $s$'' refers to the undeformed charge from which $\dkdvI_s^{(\s)}$ is constructed. The deformed current generally contains fields of different conformal weights, so this label does not assign a definite spin to the full series. Whether a non-perturbative completion restores such an interpretation remains an open question.

We conclude with a convenient relation between the corrections generated by the seeds $J_2=T$ and $J_{\s+1}=W$.

\begin{proposition}\label{prop: relation M fields}
Let $M_{s,n}$ be defined by~\eqref{eq: Ms1} and~\eqref{eq: Msn recursion}, with generating field $W=J_{\s+1}$. Then, for every $n\geq1$,
\begin{equation}
    M_{1,n+1}=2nM_{\s,n}=[W^{n+1}]
    \pmod{\del}.
\end{equation}
\end{proposition}

\begin{proof}
We prove the two identities by induction, using the skew-symmetry relation
\begin{equation}
    (AB)_1+(BA)_1=\del(AB)_2
    \pmod{\del^2}.
\end{equation}
First, we show that $M_{1,n}=[W^n]$ modulo total derivatives for all $n\geq1$. Since $J_2=T$, the initial condition gives
\begin{equation}
    M_{1,1}=\del^{-1}(TW)_1
    =\del^{-1}\del W
    =W=[W]
    \pmod{\del}.
\end{equation}
Suppose that $M_{1,k}=[W^k]$ modulo total derivatives for $1\leq k\leq n$. Substituting into~\eqref{eq: Msn recursion} yields
\begin{equation}
    M_{1,n+1}=\frac{1}{n+1}\del^{-1}\Bigg\{
    \del[W^{n+1}]
    +2n\sum_{k=1}^{n}\binom{n-1}{k-1}
    \bigl([W^k][W^{n+1-k}]\bigr)_1
    \Bigg\}
    \pmod{\del}.
\end{equation}
The binomial coefficients are invariant under $k\mapsto n+1-k$. Symmetrising the sum and applying skew-symmetry therefore gives
\begin{equation}
\begin{aligned}
    &\sum_{k=1}^{n}\binom{n-1}{k-1}
    \bigl([W^k][W^{n+1-k}]\bigr)_1\\
    &\qquad=\frac12\del\sum_{k=1}^{n}\binom{n-1}{k-1}
    \bigl([W^k][W^{n+1-k}]\bigr)_2\\
    &\qquad=\frac12\del[W^{n+1}]
    \pmod{\del^2},
\end{aligned}
\end{equation}
where the last equality follows from~\eqref{eq: W recursion}. Hence
\begin{equation}
    M_{1,n+1}
    =\frac{1}{n+1}\del^{-1}
    \bigl\{\del[W^{n+1}]+n\del[W^{n+1}]\bigr\}
    =[W^{n+1}]
    \pmod{\del}.
\end{equation}
Next, we show that $M_{\s,n}=[W^{n+1}]/(2n)$ modulo total derivatives\footnote{It looks as if when $\sigma =1$ this contradicts the previous statement. There is no contradiction. For $\s=1$, the generating field is $W=T$. Since $(TT)_2=2T$,
the recursion~\eqref{eq: W recursion} gives
\[
    [T^n]=2^{n-1}(n-1)!\,T \pmod{\del}.
\]
Thus $[T^{n+1}]=2n[T^n]$ modulo total derivatives, and the
two identities in the proposition remain consistent when
the seed labels $1$ and $\s$ coincide.}. For $n=1$, equation~\eqref{eq: Ms1} gives
\begin{equation}
    M_{\s,1}=\frac{\s}{2\s}(WW)_2
    =\frac12[W^2]
    \pmod{\del}.
\end{equation}
Suppose the identity holds for $1\leq k\leq n$. One can absorb the factors of $n$ and $k^{-1}$ to find
\begin{equation}
    M_{\s,n+1}
    =\frac{1}{n+1}\del^{-1}
    \sum_{j=0}^{n}\binom{n}{j}
    \bigl([W^{j+1}][W^{n+1-j}]\bigr)_1
    \pmod{\del}.
\end{equation}
The $j=0$ term comes from $\bigl(W[W^{n+1}]\bigr)_1$. Symmetrising under $j\mapsto n-j$ and applying skew-symmetry once more, we obtain
\begin{equation}
\begin{aligned}
    M_{\s,n+1}
    &=\frac{1}{2(n+1)}\del^{-1}
    \sum_{j=0}^{n}\binom{n}{j}
    \Bigl\{
    \bigl([W^{j+1}][W^{n+1-j}]\bigr)_1
    +\bigl([W^{n+1-j}][W^{j+1}]\bigr)_1
    \Bigr\}\\
    &=\frac{1}{2(n+1)}
    \sum_{j=0}^{n}\binom{n}{j}
    \bigl([W^{j+1}][W^{n+1-j}]\bigr)_2\\
    &=\frac{1}{2(n+1)}[W^{n+2}]
    \pmod{\del},
\end{aligned}
\end{equation}
where the final step again uses~\eqref{eq: W recursion}. Combining the two identities proves the proposition.
\end{proof}

\section{The $\sigma$-Deformed KdV Hierarchy Mutually Commutes}\label{sec: The Deformed Hierarchy Commutes}

We now prove that the $\s$-deformed KdV charges commute to all orders in $\gamma$. We give two complementary arguments. The first works directly with the recursion relations; the second reformulates them as flow equations and shows that these flows preserve commutativity. Both arguments also establish the existence of the local primitives required by the construction. The two arguments were checked in an AI audit.

Throughout this section, $V$ denotes the Virasoro vacuum module, and $V_H$ its subspace of conformal weight $H$. We write $\del V$ for the space of total derivatives and use the property $\ker\del=\C\mathbf{1}$, so that $\del$ is injective on fields of positive conformal weight. All power series are understood formally.

\subsection{Commutativity from the recursion}

The key step is an identity relating the first-order OPE poles of the component fields $M_{s,n}$. This identity ensures that the coefficient of each power of $\gamma$ in the commutator is the zero-mode of a total derivative.

\begin{proposition}\label{prop: relation first poles Ms}
Let the generating field $W$ be a quasi-primary KdV current $J_{\s+1}$. Recall that the undeformed seed fields satisfy, by definition of being KdV fields, 
\begin{equation}
    (J_{s+1}J_{r+1})_1=0\pmod{\del}
\end{equation}
for every pair $s,r$. Fix representatives of the fields $[W^n]$ satisfying~\eqref{eq: W recursion}. Then the recursion~\eqref{eq: Msn recursion}, with initial condition
\begin{equation}
    M_{s,1}=\del^{-1}(J_{s+1}W)_1\pmod{\del},
\end{equation}
admits local primitives at every order. Moreover, for every $N\geq1$,
\begin{equation}\label{eq: identiy first order pole}
\begin{aligned}
    0={}&(J_{s+1}M_{r,N})_1-(J_{r+1}M_{s,N})_1\\
    &+2(N-1)\sum_{n=1}^{N-1}\binom{N-2}{n-1}
    (M_{s,n}M_{r,N-n})_1
    \pmod{\del},
\end{aligned}
\end{equation}
where the sum is empty when $N=1$.
\end{proposition}

\begin{proof}
It is useful to absorb the combinatorial factors into the fields. Define
\begin{equation}
    A_{s,0}:=J_{s+1},
    \qquad
    A_{s,n}:=\frac{M_{s,n}}{2^{n-1}(n-1)!}
    \quad(n\geq1),
\end{equation}
and
\begin{equation}
    K_n:=\frac{[W^{n+1}]}{2^n n!}
    \quad(n\geq0).
\end{equation}
The initial condition and the recursion~\eqref{eq: Msn recursion} then take the form
\begin{equation}\label{eq: compacted recursion}
    (n+1)A_{s,n+1}
    =\del^{-1}\sum_{i+j=n}(A_{s,i}K_j)_1
    \pmod{\del},
    \qquad n\geq0,
\end{equation}
where all of the indices in the sums of this form are non-negative. The fields $K_n$ satisfy the same recursion. Indeed, equation~\eqref{eq: W recursion} becomes
\begin{equation}\label{eq: W recursion now K recursion}
    (n+1)K_{n+1}
    =\frac12\sum_{i+j=n}(K_iK_j)_2
    \pmod{\del}.
\end{equation}
Using skew-symmetry,
\begin{equation}
    (AB)_1+(BA)_1=\del(AB)_2
    \pmod{\del^2},
\end{equation}
and the symmetry of the sum under $i\leftrightarrow j$, we obtain
\begin{equation}
    \sum_{i+j=n}(K_iK_j)_1
    =\frac12\del\sum_{i+j=n}(K_iK_j)_2
    \pmod{\del^2}.
\end{equation}
Consequently,
\begin{equation}
    (n+1)K_{n+1}
    =\del^{-1}\sum_{i+j=n}(K_iK_j)_1
    \pmod{\del}.
\end{equation}
One could choose $A_{\s,n}=K_n$, consistently with Proposition~\ref{prop: relation M fields}. Define
\begin{equation}
    C_{s,r;N}:=\sum_{i+j=N}(A_{s,i}A_{r,j})_1.
\end{equation}
For $N\geq1$, the condition $C_{s,r;N}=0\pmod{\del}$ is equivalent to~\eqref{eq: identiy first order pole}. We prove this condition and the existence of the fields $A_{s,N}$ simultaneously by induction. The base case is automatic
\begin{equation}
    C_{s,r;0}=(J_{s+1}J_{r+1})_1
    =0\pmod{\del}
\end{equation}
by the assumed commutativity of the undeformed KdV currents. Now let $N\geq1$. Suppose that the fields have been constructed through order $N-1$ and that, for every pair $s,r$,
\begin{equation}\label{eq: sum of AA}
    C_{s,r;m}=\sum_{i+j=m}(A_{s,i}A_{r,j})_1
    =0\pmod{\del},
    \qquad 0\leq m<N.
\end{equation}
Setting $r=\s$ and $m=N-1$ gives
\begin{equation}
    \sum_{i+j=N-1}(A_{s,i}K_j)_1
    =0\pmod{\del}.
\end{equation}
This proves existence by induction since the primitive of the left hand side is what defines the next $A_{s,N}$ in \eqref{eq: compacted recursion}, which can be done since we've shown that it must be a total derivative.
\begin{equation}\label{eq: prim A K = m+1 A}
    \del^{-1}\sum_{i+j=m}(A_{s,i}K_j)_1
    =(m+1)A_{s,m+1}\pmod{\del},
    \qquad 0\leq m<N.
\end{equation}
It remains to prove $C_{s,r;N}=0\pmod{\del}$. Consider 
\begin{equation}
    NC_{s,r;N}
    =\sum_{i+j=N}(i+j)(A_{s,i}A_{r,j})_1,
\end{equation}
separate the terms proportional to $i$ and $j$, shift the summation indices, and use skew-symmetry. This gives
\begin{equation}\label{eq: N CsrN}
    NC_{s,r;N}
    =\sum_{i+j=N-1}(j+1)
    \bigl\{(A_{s,i}A_{r,j+1})_1
    -(A_{r,i}A_{s,j+1})_1\bigr\}
    \pmod{\del}.
\end{equation}
Differentiating~\eqref{eq: prim A K = m+1 A} yields
\begin{equation}
    (j+1)\del A_{s,j+1}
    =\sum_{u+v=j}(A_{s,u}K_v)_1
    \pmod{\del^2}.
\end{equation}
This and $(\del A\,B)_1=0$ imply that differentiating \eqref{eq: N CsrN} yields
\begin{equation}
\begin{aligned}
    N\del C_{s,r;N}
    &=\sum_{i+j=N-1}(j+1)
    \bigl\{(A_{s,i}\del A_{r,j+1})_1
    -(A_{r,i}\del A_{s,j+1})_1\bigr\}\\
    &=\sum_{i+u+v=N-1}
    \bigl\{(A_{s,i}(A_{r,u}K_v)_1)_1
    -(A_{r,u}(A_{s,i}K_v)_1)_1\bigr\}
    \pmod{\del^2},
\end{aligned}
\end{equation}
where we have interchanged $i$ and $u$ in the second term of the last sum. Application of the Jacobi Identity brings us to
\begin{equation}
\begin{aligned}
    N\del C_{s,r;N}
    &=\sum_{i+u+v=N-1}
    \bigl((A_{s,i}A_{r,u})_1K_v\bigr)_1\\
    &=\sum_{v=0}^{N-1}(C_{s,r;N-1-v}K_v)_1\\
    &=0\pmod{\del^2}.
\end{aligned}
\end{equation}
The last equality follows from the induction hypothesis: each $C_{s,r;N-1-v}$ is a total derivative, whose first-order pole in the first OPE slot vanishes.

Thus $\del C_{s,r;N}\in\del^2V$. The field $C_{s,r;N}$ has positive conformal weight, so $\ker\del=\C\mathbf{1}$ excludes an integration constant. It follows that $C_{s,r;N}\in\del V$, completing the induction.
\end{proof}
We are now in a position to show that the commutativity of the deformed charges follows immediately.
\begin{theorem}
The $\s$-deformed KdV charges defined in~\eqref{eq: dkdv definition} satisfy
\begin{equation}
    [\dkdvI_s^{(\s)},\dkdvI_r^{(\s)}]=0
\end{equation}
to all orders in $\gamma$.
\end{theorem}

\begin{proof}
In the notation of the preceding proof,
\begin{equation}
    \mathsf J_s=\sum_{n\geq0}\gamma^n A_{s,n},
    \qquad
    \dkdvI_s^{(\s)}=\sfQ{\mathsf J_s}.
\end{equation}
The cylinder zero-modes obey
\begin{equation}
    [\sfQ{A},\sfQ{B}]=\sfQ{(AB)_1}.
\end{equation}
Hence
\begin{equation}
\begin{aligned}
    [\dkdvI_s^{(\s)},\dkdvI_r^{(\s)}]
    &=[\sfQ{\mathsf J_s},\sfQ{\mathsf J_r}]\\
    &=\sum_{N\geq0}\gamma^N
    \sfQ{\sum_{i+j=N}(A_{s,i}A_{r,j})_1}.
\end{aligned}
\end{equation}
The coefficient at $N=0$ vanishes by commutativity of the undeformed charges. For every $N\geq1$, Proposition~\ref{prop: relation first poles Ms} shows that the field inside the zero-mode is a total derivative. Every coefficient of the commutator therefore vanishes.
\end{proof}

\subsection{Flow equations}

The recursion relations can also be expressed as differential equations in $\gamma$. This formulation gives a second proof of commutativity and establishes uniqueness of the deformed currents modulo total derivatives, for a fixed generating series $\mathsf K$. Recall that
\begin{align}
    \mathsf K(\gamma)
    &=\sum_{n\geq0}\frac{\gamma^n}{2^n n!}[W^{n+1}],\\
    \mathsf J_s(\gamma)
    &=J_{s+1}+\sum_{n\geq1}
    \frac{\gamma^n}{2^{n-1}(n-1)!}M_{s,n},
    \qquad
    \gamma=\frac{\alpha_\s\t^\s}{2\pi i}.
\end{align}
Write these series as
\begin{equation}
    \mathsf J_s=\sum_{n\geq0}j_{s,n}\gamma^n,
    \qquad
    \mathsf K=\sum_{n\geq0}k_n\gamma^n.
\end{equation}
Their coefficients have conformal weights such that
\begin{equation}\label{eq: ao coefficient grading}
    j_{s,n}\in V_{s+1+n(\s-1)},
    \qquad
    k_n\in V_{\s+1+n(\s-1)}.
\end{equation}
For the flow equations, we fix the representatives of $[W^n]$ by imposing the second-order-pole recursion~\eqref{eq: W recursion}.

\begin{proposition}\label{prop: ao flows}
With this choice of representatives, the recursion~\eqref{eq: W recursion} is formally equivalent to the flow equation
\begin{equation}\label{eq: ao K flow}
    \frac{\partial\mathsf K}{\partial\gamma}
    =\frac12(\mathsf K\mathsf K)_2,
    \qquad \mathsf K(0)=W.
\end{equation}
For this fixed $\mathsf K$, equations~\eqref{eq: Ms1} and~\eqref{eq: Msn recursion} are equivalent to
\begin{equation}\label{eq: djdgamma = }
    \frac{\partial\mathsf J_s}{\partial\gamma}
    =\del^{-1}(\mathsf J_s\mathsf K)_1
    \pmod{\del},
    \qquad \mathsf J_s(0)=J_{s+1}.
\end{equation}
Moreover, $\mathsf K$ itself satisfies the primitive flow
\begin{equation}\label{eq: ao K primitive flow}
    \frac{\partial\mathsf K}{\partial\gamma}
    =\del^{-1}(\mathsf K\mathsf K)_1
    \pmod{\del}.
\end{equation}
\end{proposition}

\begin{proof}
Differentiating $\mathsf J_s$ term by term gives
\begin{equation}
    \frac{\partial\mathsf J_s}{\partial\gamma}
    =\sum_{n\geq1}
    \frac{\gamma^{n-1}}{2^{n-1}(n-1)!}\,nM_{s,n}.
\end{equation}
On the other hand, bilinearity of the first-order pole yields
\begin{equation}
\begin{aligned}
    (\mathsf J_s\mathsf K)_1
    ={}&\sum_{n\geq1}
    \frac{\gamma^{n-1}}{2^{n-1}(n-1)!}
    (J_{s+1}[W^n])_1\\
    &+\sum_{k,\ell\geq1}
    \frac{\gamma^{k+\ell-1}}
    {2^{k+\ell-2}(k-1)!(\ell-1)!}
    (M_{s,k}[W^\ell])_1.
\end{aligned}
\end{equation}
Setting $n=k+\ell$ in the second sum gives
\begin{equation}
\begin{aligned}
    (\mathsf J_s\mathsf K)_1
    ={}&(J_{s+1}W)_1\\
    &+\sum_{n\geq2}
    \frac{\gamma^{n-1}}{2^{n-1}(n-1)!}
    \Bigg\{(J_{s+1}[W^n])_1
    +2(n-1)\sum_{k=1}^{n-1}\binom{n-2}{k-1}
    (M_{s,k}[W^{n-k}])_1\Bigg\}.
\end{aligned}
\end{equation}
Each expression in braces is the input whose primitive defines $nM_{s,n}$ in~\eqref{eq: Msn recursion}. Together with the initial condition~\eqref{eq: Ms1}, this gives
\begin{equation}
    (\mathsf J_s\mathsf K)_1
    =\sum_{n\geq1}
    \frac{\gamma^{n-1}}{2^{n-1}(n-1)!}\,n\del M_{s,n}
    =\del\frac{\partial\mathsf J_s}{\partial\gamma}
    \pmod{\del^2}.
\end{equation}
Taking primitives order by order proves~\eqref{eq: djdgamma = }. Conversely, coefficient comparison in that flow equation recovers~\eqref{eq: Ms1} and~\eqref{eq: Msn recursion} modulo total derivatives. For $\mathsf K$, we have
\begin{equation}
\begin{aligned}
    \frac{\partial\mathsf K}{\partial\gamma}
    &=\sum_{n\geq0}\frac{\gamma^n}{2^{n+1}n!}[W^{n+2}],\\
    \frac12(\mathsf K\mathsf K)_2
    &=\sum_{n\geq0}\frac{\gamma^n}{2^{n+1}n!}
    \sum_{a=0}^{n}\binom{n}{a}
    \bigl([W^{a+1}][W^{n-a+1}]\bigr)_2.
\end{aligned}
\end{equation}
Equality of the coefficients is precisely the exact representative version of~\eqref{eq: W recursion}, proving~\eqref{eq: ao K flow}. Finally, skew-symmetry gives
\begin{equation}
    (\mathsf K\mathsf K)_1
    =\frac12\del(\mathsf K\mathsf K)_2
    -\frac14\del^2(\mathsf K\mathsf K)_3
    +\frac1{12}\del^3(\mathsf K\mathsf K)_4-\cdots.
\end{equation}
A local primitive therefore exists, and
\begin{equation}
    \del^{-1}(\mathsf K\mathsf K)_1
    =\frac12(\mathsf K\mathsf K)_2
    \pmod{\del}.
\end{equation}
Combining this with~\eqref{eq: ao K flow} proves~\eqref{eq: ao K primitive flow}.
\end{proof}
The local-primitive operation entering our flows is the operation used in Dijkgraaf’s contact-term construction \cite{Dijkgraaf:1996iy}, subsequently formulated in vertex-algebra language in \cite{Gui:2025kbd}. In Dijkgraaf’s notation, \(\nabla_B A=\partial^{-1}(AB)_1\), so our current flow reads \(\partial_\gamma\mathsf J_s=\nabla_{\mathsf K}\mathsf J_s\) modulo total derivatives. The contact-algebra framework assumes a suitable closed collection of fields whose first-order OPE poles are total derivatives for Dijkgraaf's integrability condition. Here we do not assume that the individual coefficient fields satisfy this condition pairwise. Instead, for the prescribed generating series \(\mathsf K\), we prove inductively that the complete coefficient of each power of \(\gamma\) in \((\mathsf J_s\mathsf K)_1\) admits a local primitive, and that the resulting formal charges commute at a common deformation parameter. Our construction uses this established operation while proving the existence and commutativity required for the deformed hierarchy. The following lemma is useful to explain why the flows preserve commutativity.

\begin{lemma}\label{lem: ao commutativity}
Let $A,B,K\in V$ be homogeneous fields of weights $h_A,h_B,h_K\geq2$. Suppose that their first-order poles are total derivatives,
\begin{equation}
    (AB)_1=\del H,
    \qquad (AK)_1=\del F,
    \qquad (BK)_1=\del G,
\end{equation}
with homogeneous local primitives $H,F,G$. Then
\begin{equation}
    (FB)_1+(AG)_1\in\del V.
\end{equation}
The same conclusion holds coefficient by coefficient for formal series, or for series truncated at any order, whose coefficient weights are respectively
\begin{equation}
    h_A+n(\s-1),
    \qquad h_B+n(\s-1),
    \qquad h_K+n(\s-1).
\end{equation}
\end{lemma}

\begin{proof}
Acting with derivatives and considering the Jacobi identity, one finds
\begin{equation}
\begin{aligned}
    \del\bigl((AG)_1-(BF)_1\bigr)
    &=(A(\del G))_1-(B(\del F))_1\\
    &=(A(BK)_1)_1-(B(AK)_1)_1\\
    &=((AB)_1K)_1\\
    &=((\del H)K)_1=0.
\end{aligned}
\end{equation}
The field $(AG)_1-(BF)_1$ has weight $h_A+h_B+h_K-3>0$. Injectivity of $\del$ at positive weight therefore implies
\begin{equation}
    (AG)_1=(BF)_1.
\end{equation}
Since $(BF)_1=-(FB)_1\pmod{\del}$ by skew-symmetry, the result follows. For formal or truncated series, the same calculation applies coefficient by coefficient. The coefficient of $\gamma^n$ in $(AG)_1-(BF)_1$ has weight
\begin{equation}
    h_A+h_B+h_K-3+n(\s-1)>0,
\end{equation}
so the injectivity argument remains valid at every order under consideration.
\end{proof}

We can now establish existence, uniqueness, and commutativity directly from the flow equations.

\begin{proposition}\label{prop: ao hierarchy}
Recall that the undeformed KdV currents have first-order OPE poles in $\del V$, and fix the full series $\mathsf K$ by the exact flow~\eqref{eq: ao K flow}. Then:
\begin{enumerate}
    \item All local primitives required by~\eqref{eq: Ms1} and~\eqref{eq: Msn recursion} exist.
    \item The currents $\mathsf J_s$ satisfying~\eqref{eq: djdgamma = }, with initial values $\mathsf J_s(0)=J_{s+1}$ and grading~\eqref{eq: ao coefficient grading}, are uniquely determined modulo total derivatives.
    \item Every pair of fields in $\{\mathsf K,\mathsf J_1,\mathsf J_3,\ldots\}$ commutes locally:
    \begin{equation}
        (AB)_1\in\del V[\![\gamma]\!],
        \qquad A,B\in\{\mathsf K,\mathsf J_1,\mathsf J_3,\ldots\}.
    \end{equation}
\end{enumerate}
Here $\del V[\![\gamma]\!]$ denotes the formal series whose coefficients are total derivatives.
\end{proposition}

\begin{proof}
We prove existence and local commutativity simultaneously by induction on the power of $\gamma$, keeping $\mathsf K$ fixed throughout. We then prove uniqueness.

At order zero, $j_{s,0}=J_{s+1}$ and $k_0=W$ are undeformed KdV currents, so every pair commutes locally. Suppose that the currents $\mathsf J_s$ have been constructed through degree $N$, satisfy their flow equations through degree $N-1$, and commute locally with one another and with $\mathsf K$ through degree $N$. For $N=0$, the condition on the flow equations is empty. By the induction hypothesis, there are homogeneous local fields $U_{s,n}$ such that
\begin{equation}
    (\mathsf J_s\mathsf K)_1
    =\sum_{n=0}^{N}\del U_{s,n}\gamma^n
    +O(\gamma^{N+1}).
\end{equation}
The coefficient of $\gamma^N$ depends only on the fields already constructed. We may therefore define
\begin{equation}\label{eq: ao next J}
    j_{s,N+1}=\frac{U_{s,N}}{N+1}\pmod{\del}.
\end{equation}
This extends the flow equation through degree $N$. The primitive $U_{s,N}$ has weight $s+1+(N+1)\\(\s-1)>0$, so the prescribed grading is preserved. The corresponding coefficient of $\mathsf K$ is already fixed exactly
\begin{equation}
    k_{N+1}
    =\frac{1}{2(N+1)}\sum_{a=0}^{N}(k_a k_{N-a})_2.
\end{equation}
By~\eqref{eq: ao K primitive flow}, this coefficient also satisfies the primitive flow modulo total derivatives. Thus every member of the extended family satisfies its primitive flow through degree $N$.

To propagate commutativity, choose any two members $A,B$ of this family. Modulo $\gamma^{N+1}$, the induction hypothesis gives
\begin{equation}
    (A\mathsf K)_1=\del F,
    \qquad (B\mathsf K)_1=\del G,
    \qquad (AB)_1=\del H,
\end{equation}
while the extended flow equations give
\begin{equation}
    \frac{\partial A}{\partial\gamma}= F,
    \qquad
    \frac{\partial B}{\partial\gamma}= G
    \pmod{\del,\gamma^{N+1}}.
\end{equation}
Derivative corrections do not affect the calculation modulo $\del$, since
\begin{equation}
    ((\del X)B)_1=0,
    \qquad (A(\del Y))_1=\del(AY)_1.
\end{equation}
Applying Lemma~\ref{lem: ao commutativity} to these truncated series therefore yields
\begin{equation}
\begin{aligned}
    \frac{\partial}{\partial\gamma}(AB)_1
    &=\left(\frac{\partial A}{\partial\gamma}B\right)_1
    +\left(A\frac{\partial B}{\partial\gamma}\right)_1\\
    &=(FB)_1+(AG)_1\\
    &=0\pmod{\del,\gamma^{N+1}}.
\end{aligned}
\end{equation}
Writing
\begin{equation}
    (AB)_1=\sum_{n\geq0}C_n\gamma^n
\end{equation}
and comparing the coefficient of $\gamma^N$ gives
\begin{equation}
    (N+1)C_{N+1}=0\pmod{\del}.
\end{equation}
Hence $C_{N+1}\in\del V$, and every pair commutes locally through degree $N+1$. This closes the induction. Proposition~\ref{prop: ao flows} then also establishes the existence of every local primitive in the original recursion.

For uniqueness, let $\mathsf J_s$ and $\widetilde{\mathsf J}_s$ be two solutions for the same fixed $\mathsf K$, with the same initial value and coefficient grading. Suppose inductively that
\begin{equation}
    j_{s,n}-\widetilde j_{s,n}\in\del V,
    \qquad 0\leq n\leq N.
\end{equation}
This holds at $N=0$. A derivative in the first OPE slot has no first-order pole, so the $\gamma^N$ term in $((\mathsf J_s-\widetilde{\mathsf J}_s)\mathsf K)_1$ vanishes. The degree-$N$ inputs to the two primitive flows therefore agree exactly. Their homogeneous primitives also agree: their difference lies in $\ker\del$, while their weight $s+1+(N+1)(\s-1)$ is positive. Comparing the flow equations at degree $N$ now gives
\begin{equation}
    (N+1)\bigl(j_{s,N+1}-\widetilde j_{s,N+1}\bigr)
    =0\pmod{\del}.
\end{equation}
Induction proves uniqueness modulo total derivatives. This uniqueness statement holds for the prescribed full series $\mathsf K$; it does not permit arbitrary changes to its derivative terms.
\end{proof}

Since the action of extracting cylinder zero-modes, $\sfQ{\cdot}$, annihilates total derivatives, the proposition gives a second proof that
\begin{equation}
\begin{aligned}
    [\dkdvI_s^{(\s)},\dkdvI_r^{(\s)}]
    &=[\sfQ{\mathsf J_s},\sfQ{\mathsf J_r}]\\
    &=\sfQ{(\mathsf J_s\mathsf J_r)_1}=0
\end{aligned}
\end{equation}
to all orders in $\gamma$. Some remarks are in order.

\medskip
\noindent\emph{Recovery of the deformed Hamiltonian.}
The case $s=1$ recovers the Hamiltonian obtained from the modular transformation. To see this directly from the flow, set
\begin{equation}
    \widehat{\mathsf J}_1:=T+\gamma\mathsf K.
\end{equation}
This field has the required initial value and coefficient grading. Since $(T\mathsf K)_1=\del\mathsf K$, equation~\eqref{eq: ao K primitive flow} gives
\begin{equation}
\begin{aligned}
    \del^{-1}(\widehat{\mathsf J}_1\mathsf K)_1
    &=\mathsf K+\gamma\del^{-1}(\mathsf K\mathsf K)_1\\
    &=\mathsf K+\gamma\frac{\partial\mathsf K}{\partial\gamma}\\
    &=\frac{\partial\widehat{\mathsf J}_1}{\partial\gamma}
    \pmod{\del}.
\end{aligned}
\end{equation}
Uniqueness therefore implies
\begin{equation}
    \mathsf J_1= T+\gamma\mathsf K\pmod{\del},
    \qquad
    M_{1,n}=[W^n]\pmod{\del},
\end{equation}
in agreement with Proposition~\ref{prop: relation M fields}.

\medskip\noindent
\emph{Dependence on the local representatives.}
The uniqueness statement above fixes the full local series
$\mathsf K$, rather than only the integrated Hamiltonian
$\dkdvI_1^{(\s)}=\I_1+\gamma\sfQ{\mathsf K}$.
Indeed, a coefficientwise change
\begin{equation}
    \widetilde{\mathsf K}
    =\mathsf K+\del\mathsf X,
    \qquad \mathsf X(0)=0,
\end{equation}
with the prescribed grading leaves this Hamiltonian
unchanged, but modifies the primitive operation according to
\begin{equation}
    \del^{-1}(A\widetilde{\mathsf K})_1
    =
    \del^{-1}(A\mathsf K)_1+(A\mathsf X)_1.
\end{equation}
Thus the primitive flow
\eqref{eq: ao K primitive flow} is preserved modulo total
derivatives precisely when
\begin{equation}
    (\mathsf K\mathsf X)_1=0\pmod{\del}.
\end{equation}
Under this condition, the existence and commutativity
argument also applies to the modified generating series.
For the original currents $\mathsf J_s$ to remain solutions
of the modified versions of the flows \eqref{eq: djdgamma = }, one additionally needs
\begin{equation}
    (\mathsf J_s\mathsf X)_1=0\pmod{\del}
    \qquad\text{for every }s.
\end{equation}
These conditions hold automatically when
$\mathsf X=\del\mathsf Y$; in that case uniqueness implies
that the currents agree modulo total derivatives and
the integrated charges are unchanged.

Throughout this work, we fix $\mathsf K$ by the exact
recursion \eqref{eq: W recursion}. Together with the
prescribed initial currents and grading, this determines
the higher charges uniquely. We do not claim uniqueness
from the integrated Hamiltonian alone, nor do we classify
other compatible choices of local representatives -- we leave this topic for future work.

\medskip
\noindent\emph{Relation to boost flows.} Finally, we remark on a formal similarity that the local flow equations have with boost deformations in the context of integrable spin-chains. If one considers the Noether charges associated to the local fields on the plane, and not the cylinder charges, then we find an equation with the same structure as in \cite{Bargheer:2008jt,Bargheer:2009xy}. Suppose that $(\mathsf J_s \mathsf K)_1 = \del \mathsf F_s$. Then the local flow equation \eqref{eq: djdgamma = } reads
\begin{equation}
    \frac{\partial\mathsf J_s}{\partial\gamma}
    =\mathsf F_s\pmod{\del}.
\end{equation}
Now consider the action of extracting a planar charge. For a local field $A(z)$, the planar charge is 
\begin{equation}
    \mathsf P (A) = \oint \frac{dz}{2\pi i }A(z).
\end{equation}
Then acting on the flow equation with $\mathsf P(\cdot)$ gives
\begin{equation}
    \frac{\partial\mathsf P (\mathsf J_s)}{\partial\gamma}
    =\mathsf P (\mathsf F_s).
\end{equation}
One can construct a boost operator, weighted by position, in the same spirit as \cite{Bargheer:2008jt}, 
\begin{equation}
    \mathcal{X} [\mathsf K] = \oint \frac{dz}{2\pi i}z \mathsf K(z),
\end{equation}
which can be shown to give a boost equation
\begin{equation}
\frac{\partial\mathsf P (\mathsf J_s)}{\partial\gamma}
     = [\mathcal  X [\mathsf K], \mathsf P ( \mathsf J_s)].
\end{equation}
Although this equation is similar to that which appears in \cite{Bargheer:2008jt}, it is worth noting that it is not literally the same thing. Boost operators are normally constructed as integrals on the real line in full bulk theories (or a sum over every site in the spin chain theory), whereas our construction is focusing only on the chiral side of a CFT, where the integrals are done on a closed contour. Furthermore, an attempt to construct the same equation for the cylinder charges leads to an obstruction -- one needs to introduce a cut on the cylinder since the integrand is not periodic. 

Nevertheless, it would be interesting to pursue the consequences of this equation, and discover any links to the statistical physics context. This is left to future work.

\section{The $\sigma$-Deformed KdV Hierarchy in the Free Fermion Model}\label{sec: The Deformed Hierarchy in FF}

We now apply the construction to the $c=1/2$ free fermion model. In this theory, the KdV currents admit a bilinear realisation in terms of a free fermion~\cite{Downing:2021mfw}, and the deformed charges can be written explicitly as linear combinations of the undeformed KdV charges. The TBA equation is trivial, so the solution is explicit, allowing us to compare the two descriptions to all orders in the deformation parameter. We first derive the asymptotic expansion of the ground-state eigenvalues from the TBA, then recover the same coefficients from the local recursion.

We will use two normalisations. In the BLZ normalisation, the coefficient of the leading power of the stress tensor is one:
\begin{equation}
    \I_{2n-1}=\sfQ{J_{2n}},
    \qquad J_{2n}=(T^n)+\cdots,
\end{equation}
where $(T^n)$ denotes the right-nested normal-ordered product of $n$ stress tensors. In the bilinear normalisation, the currents are given in terms of the free Majorana fermion field $\psi(w)$ of conformal weight $1/2$.
\begin{equation}
    \I_{2n-1}^{\mathrm{bil}}=\sfQ{J_{2n}^{\mathrm{bil}}},
    \qquad
    J_{2n}^{\mathrm{bil}}
    =-\frac12(\psi\del^{2n-1}\psi)
    \pmod{\del}.
\end{equation}
The charges in these two conventions are related by
\begin{equation}
    \I_s^{\mathrm{bil}}=N_s\I_s.
\end{equation}
We determine the conversion factor $N_s$ below, when matching the two expansions.

\subsection{Ground-state expansion from the TBA}

Consider a GGE with a single spin-$\s$ KdV charge inserted,
\begin{equation}
    \Tr\!\left(q^{L_0-\frac{c}{24}}e^{\alpha_\s\I_\s}\right),
\end{equation}
with the charge in BLZ normalisation. Throughout this section, $s$ and $\s$ are positive odd integers. Denote the undeformed ground-state eigenvalues in the Neveu--Schwarz and Ramond sectors by $I_s^+$ and $I_s^-$, respectively:
\begin{equation}
\begin{aligned}
    \I_s\ket{0}&=I_s^+\ket{0},\\
    \I_s\ket{\tfrac{1}{16}}&=I_s^-\ket{\tfrac{1}{16}}.
\end{aligned}
\end{equation}
The undeformed pseudoenergy is $\epsilon(\theta)=e^\theta$, giving
\begin{equation}
\begin{aligned}
    I_s^\pm
    &=B_s\int_{\R}d\theta\,e^{s\theta}
    \log\bigl(1\pm e^{-e^\theta}\bigr)\\
    &=\begin{cases}
        B_s\Gamma(s)(1-2^{-s})\zeta(s+1),&\text{Neveu--Schwarz},\\
        -B_s\Gamma(s)\zeta(s+1),&\text{Ramond}.
    \end{cases}
\end{aligned}
\end{equation}
The normalisation constant follows from~\cite[(4.36)]{Nigro:2007tz}:
\begin{equation}
    B_s=
    \frac{(-1)^{\frac{s+1}{2}}2^{s-1}
    \Gamma\!\left(\frac{s+3}{2}\right)
    \Gamma\!\left(\frac{3s+2}{2}\right)}
    {3^{\frac{s+1}{2}}\pi^{s+\frac32}\Gamma(2s)}.
\end{equation}
The spin-$\s$ insertion adds a driving term to the mirror-channel TBA, so that
\begin{equation}
    \epsilon(\theta)=e^\theta+\kappa_\s e^{\s\theta}.
\end{equation}
We will later fix the relation between $\kappa_\s$ and $\alpha_\s$ by matching the first-order term in the deformed Hamiltonian. The deformed ground-state eigenvalues are
\begin{equation}
    I_s^\pm(\kappa_\s)
    =B_s\int_{\R}d\theta\,e^{s\theta}
    \log\bigl(1\pm e^{-e^\theta-\kappa_\s e^{\s\theta}}\bigr).
\end{equation}
We take $\kappa_\s\geq0$ so that the integral converges, and compare the small-coupling expansion of this integral with the formal series defined by the local recursion. We can perform the change of variables $x=e^\theta$, expand in $\kappa_\s$, and evaluate the coefficients, so as to obtain
\begin{equation}\label{eq: Is FF TBA (kappa sigma)}
    I_s^\pm(\kappa_\s)
    \sim\sum_{m=0}^{\infty}C_m I_{s_m}^\pm\kappa_\s^m,
    \qquad s_m:=s+m(\s-1),
\end{equation}
where
\begin{equation}
    C_m
    =\frac{(-1)^m\Gamma(s+\s m)}{m!\,\Gamma(s_m)}
    \frac{B_s}{B_{s_m}}
    =(-1)^m\binom{s_m+m-1}{m}\frac{B_s}{B_{s_m}}.
\end{equation}
The coefficients are the same in both sectors. They suggest the formal operator series
\begin{equation}\label{eq: Deformed KdV Charge Free Fermion TBA.}
    \I_s^{(\s)}(\kappa_\s)
    :=\sum_{m=0}^{\infty}
    (-1)^m\binom{s_m+m-1}{m}
    \frac{B_s}{B_{s_m}}\,
    \kappa_\s^m\I_{s_m}.
\end{equation}
We now derive this operator series independently from the local recursion.

\subsection{Evaluation of the local recursion}

The recursion is simplest to evaluate in the bilinear normalisation. We therefore take the seed and generating fields to be
\begin{equation}
    J_{s+1}\to J_{s+1}^{\mathrm{bil}},
    \qquad W \to J_{\s+1}^{\mathrm{bil}}.
\end{equation}
The bilinear currents obey the OPE identities, what is identified as a pre-lie algebra in \cite{Dijkgraaf:1996iy}
\begin{equation}\label{eq: prelie algebra}
\begin{aligned}
    \bigl(J_{2n}^{\mathrm{bil}}J_{2m}^{\mathrm{bil}}\bigr)_1
    &=0\pmod{\del},\\
    \del^{-1}\bigl(J_{2n}^{\mathrm{bil}}J_{2m}^{\mathrm{bil}}\bigr)_1
    &=(2n-1)J_{2n+2m-2}^{\mathrm{bil}}
    \pmod{\del},\\
    \bigl(J_{2n}^{\mathrm{bil}}J_{2m}^{\mathrm{bil}}\bigr)_2
    &=(2n+2m-2)J_{2n+2m-2}^{\mathrm{bil}}
    \pmod{\del}.
\end{aligned}
\end{equation}
Thus the fields appearing in the recursion remain bilinear modulo total derivatives. Applying the second-order-pole recursion \eqref{eq: W recursion} gives
\begin{equation}
    [W^n]
    =
    \frac{2^{n-1}(\s n)!}
    {n\bigl(n(\s-1)+1\bigr)!}
    J_{n(\s-1)+2}^{\mathrm{bil}}
    \pmod{\del},
    \qquad n\geq1.
\end{equation}
Similarly, the mixed recursion \eqref{eq: Msn recursion} yields
\begin{equation}
    M_{s,n}
    =
    \frac{s\,2^{n-1}(s+\s n-1)!}
    {n\bigl(s+n(\s-1)\bigr)!}
    J_{s+n(\s-1)+1}^{\mathrm{bil}}
    \pmod{\del},
    \qquad n\geq1.
\end{equation}
At $s=1$, these expressions agree, as required by $M_{1,n}=[W^n]\pmod{\del}$.

Substituting into the definition of the deformed charges and including the undeformed term, we find
\begin{equation}
    \dkdvI_s^{(\s),\mathrm{bil}}
    =\sum_{m=0}^{\infty}
    \frac{s}{s_m}\binom{s_m+m-1}{m}
    \left(\frac{\alpha_\s^{\mathrm{bil}}\t^\s}{2\pi i}\right)^m
    \I_{s_m}^{\mathrm{bil}}.
\end{equation}
Here $\alpha_\s^{\mathrm{bil}}$ is the chemical potential conjugate to $\I_\s^{\mathrm{bil}}$. To compare with~\eqref{eq: Deformed KdV Charge Free Fermion TBA.}, we must convert both the charges and the chemical potential to BLZ normalisation.

\subsection*{Matching the normalisations}

In the bilinear convention, the undeformed ground-state eigenvalues are
\begin{equation}
    I_{s,\mathrm{bil}}^\pm
    =\begin{cases}
        -\frac12(2^{-s}-1)\zeta(-s),&\text{Neveu--Schwarz},\\
        -\frac12\zeta(-s),&\text{Ramond}.
    \end{cases}
\end{equation}
Comparing these with the BLZ eigenvalues determines the conversion factor:
\begin{equation}\label{eq:relation bil to BLZ}
    \I_s^{\mathrm{bil}}=N_s\I_s,
    \qquad I_{s,\mathrm{bil}}^\pm=N_sI_s^\pm,
\end{equation}
where
\begin{equation}
\begin{aligned}
    N_s
    &=\frac{s\,3^{\frac{s+1}{2}}\sqrt{\pi}\,\Gamma(2s)}
    {2^{2s}\Gamma\!\left(\frac{s+3}{2}\right)
    \Gamma\!\left(\frac{3s+2}{2}\right)}\\
    &=\frac{(-1)^{\frac{s+1}{2}}s}{(2\pi)^{s+1}B_s}.
\end{aligned}
\end{equation}
The chemical potentials must be rescaled so that the GGE insertion is unchanged:
\begin{equation}
    \alpha_\s\I_\s
    =\alpha_\s^{\mathrm{bil}}\I_\s^{\mathrm{bil}},
    \qquad \alpha_\s^{\mathrm{bil}}=\frac{\alpha_\s}{N_\s}.
\end{equation}
Dividing the bilinear deformed charge by $N_s$ gives a charge whose undeformed term is $\I_s$:
\begin{equation}
    \dkdvI_s^{(\s)}
    =\sum_{m=0}^{\infty}
    \frac{s}{s_m}\binom{s_m+m-1}{m}
    \frac{N_{s_m}}{N_s}
    \left(\frac{\alpha_\s\t^\s}{2\pi iN_\s}\right)^m
    \I_{s_m}.
\end{equation}

To identify $\kappa_\s$, compare the first-order terms for $s=1$. Since $N_1=1$, the local recursion gives
\begin{equation}
    \dkdvI_1^{(\s)}
    =\I_1+\frac{\alpha_\s\t^\s}{2\pi i}\I_\s
    +O(\alpha_\s^2),
\end{equation}
whereas the series suggested by the TBA gives
\begin{equation}
    \I_1^{(\s)}(\kappa_\s)
    =\I_1-\s\frac{B_1}{B_\s}\kappa_\s\I_\s
    +O(\kappa_\s^2).
\end{equation}
Using $B_1=-1/(4\pi^2)$, equality of the first-order terms fixes
\begin{equation}
    \kappa_\s
    =-\frac{2\pi iB_\s}{\s}\alpha_\s\t^\s.
\end{equation}
With this identification, the normalisation factors satisfy
\begin{equation}
    \frac{N_{s_m}}{N_s}
    \left(\frac{\alpha_\s\t^\s}{2\pi iN_\s}\right)^m
    =(-1)^m\frac{s_m}{s}\frac{B_s}{B_{s_m}}\kappa_\s^m.
\end{equation}
The factors $s/s_m$ cancel, leaving
\begin{equation}
    \dkdvI_s^{(\s)}
    =\sum_{m=0}^{\infty}
    (-1)^m\binom{s_m+m-1}{m}
    \frac{B_s}{B_{s_m}}\,
    \kappa_\s^m\I_{s_m}
    =\I_s^{(\s)}(\kappa_\s).
\end{equation}
Thus the local recursion reproduces the formal operator series~\eqref{eq: Deformed KdV Charge Free Fermion TBA.} coefficient by coefficient. In particular, its ground-state eigenvalues agree with the TBA expansion~\eqref{eq: Is FF TBA (kappa sigma)} in both sectors.

This expression also makes commutativity immediate at $c=1/2$: every deformed charge is a formal linear combination of the mutually commuting undeformed KdV charges. Hence
\begin{equation}
    [\dkdvI_s^{(\s)},\dkdvI_r^{(\s)}]=0.
\end{equation}
The same argument applies even when the charges are constructed using different generating fields, and when we choose different chemical potentials
\begin{equation}
    [\dkdvI_s^{(\s)}(\alpha_\s),\dkdvI_r^{(\rho)}(\alpha_\rho)]=0
\end{equation}
for arbitrary positive odd $s,r,\s,\rho$, as an identity of formal series in the respective deformation parameters. This is a special feature of $c=1/2$. The same property holds at $c=-2$ \cite{Karimi:2026tvh}, since the KdV fields satisfy the same relations \eqref{eq: prelie algebra}, and the TBA is again trivial.

\section{The $\sigma$-Deformed KdV Hierarchy in the Lee-Yang Model}\label{sec: The Deformed Hierarchy in LY}

The free fermion model provides an explicit check of the construction, but its bilinear realisation simplifies both the deformed charges and the TBA. We now turn to the Lee-Yang minimal model, where the TBA contains a non-trivial integral term. Our aim is to test whether the higher deformed charges obtained from the local recursion have the same TBA interpretation as the deformed Hamiltonian. We compare the two descriptions through cubic order in the deformation parameter, using the ground state and three low-lying states.

\subsection{Setup}

The Lee-Yang model $M(2,5)$ has central charge $c=-22/5$ and two Virasoro primary fields, of weights $h=0$ and $h=-1/5$. The spin-$3$ KdV charge is null in this theory, so it does not generate a non-trivial deformation. We therefore consider the GGE with a spin-$5$ insertion,
\begin{equation}
    \Tr\!\left(q^{L_0-\frac{c}{24}}e^{\alpha\I_5}\right),
    \qquad \alpha:=\alpha_5,
\end{equation}
and take $W=J_6$ as the generating field. We retain the charge normalisation of~\eqref{eq: dkdv definition} and write
\begin{equation}
    \gamma=\frac{\alpha\t^5}{2\pi i}.
\end{equation}
For $s=1,5,7$, the local recursion gives
\begin{equation}
    \dkdvI_s^{(5)}
    =\I_s+\sum_{m=1}^{3}b_m\gamma^m\J_{s,s+4m}
    +O(\gamma^4),
    \qquad b_m=\frac{1}{2^{m-1}(m-1)!}.
\end{equation}
Here $\J_{s,s+4m}$ denotes the cylinder zero-mode of the corresponding correction field. We evaluate these charges on the following states in the $h=-1/5$ module:
\begin{equation}
    \ket{h},\qquad L_{-1}\ket{h},\qquad
    L_{-2}\ket{h},\qquad L_{-3}\ket{h},
    \qquad h=-\frac15.
\end{equation}
We label them by $\{0\},\{1\},\{2\},\{3\}$, respectively. After quotienting by null states, each of these levels is one-dimensional, so there are no other states to worry about. Denoting the eigenvalue of the deformed charge by $\mathcal I_s^{(5)}\{i\}$, we have
\begin{equation}
    \mathcal I_s^{(5)}\{i\}
    =I_s\{i\}+\sum_{m=1}^{3}b_m\gamma^m J_{s,s+4m}\{i\}
    +O(\gamma^4),
    \qquad i=0,1,2,3.
\end{equation}
We first compare ratios of ground-state to excited-state coefficients:
\begin{equation}
    R_s^{(i)}:=\frac{I_s\{0\}}{I_s\{i\}},
    \qquad
    R_{s,s+4m}^{(i)}:=\frac{J_{s,s+4m}\{0\}}{J_{s,s+4m}\{i\}},
    \qquad i,m=1,2,3.
\end{equation}
These ratios are independent of the overall normalisation of each component charge and of the conversion between deformation parameters. Once the ratios agree, we fix that conversion and compare the raw coefficients (not ratios) in the case of the ground state. The ground-state TBA equation is
\begin{equation}
    \epsilon(\theta)
    =e^\theta+\kappa e^{5\theta}
    -\int_\R\varphi(\theta-\theta')
    \log\bigl(1+e^{-\epsilon(\theta')}\bigr)
    \frac{d\theta'}{2\pi},
\end{equation}
where
\begin{equation}
    S(\theta)=\frac{\sinh\theta+i\sin(\pi/3)}
    {\sinh\theta-i\sin(\pi/3)},
    \qquad
    \varphi(\theta)=-i\,\partial_\theta\log S(\theta)
    =-\frac{4\sqrt3\cosh\theta}{1+2\cosh(2\theta)}.
\end{equation}
This equation was obtained by performing a mirror transformation on the TBA equation in section 5 of \cite{Downing:2024nfb}. At $\kappa=0$, its solution determines the undeformed ground-state eigenvalues through
\begin{equation}
    I_{2n-1}\{0\}
    =K_n\int_\R e^{(2n-1)\theta}
    \log\bigl(1+e^{-\epsilon(\theta)}\bigr)
    \frac{d\theta}{2\pi}.
\end{equation}
The normalisation constants follow from~\cite{Bazhanov:1996aq}:
\begin{equation}
    K_n=\frac{3\,5^{n-1}n!}{2^{3n-3}\pi^n}
    \frac{\Gamma\!\left(\frac{2n+2}{3}\right)}
    {\Gamma\!\left(\frac{10n-5}{6}\right)}
    \left[\frac{\Gamma(2/3)}{\Gamma(1/6)}\right]^{2n-1}
    \sin\!\left(\frac{2\pi(n+1)}{3}\right).
\end{equation}
We test the proposal that the same integral, evaluated with the deformed pseudoenergy, gives the eigenvalues of the deformed charges:
\begin{equation}
    \mathcal I_{2n-1}^{(5),\mathrm{TBA}}\{0\}
    :=K_n\int_\R e^{(2n-1)\theta}
    \log\bigl(1+e^{-\epsilon(\theta)}\bigr)
    \frac{d\theta}{2\pi},
    \qquad \kappa\geq0,
\end{equation}
Where we emphasise that the expression above is dependant on $\kappa$. For the Hamiltonian, this relation was studied in~\cite{Downing:2024nfb}, where the TBA energies were compared with the Hamiltonian obtained from the modular transformation. Here we extend the comparison to the higher charges.

For an excited state, additional source terms are specified by complex rapidities $\vartheta_a$ and their conjugates, which constitute poles in the integrand of the convolution term. The TBA equation becomes\footnote{This equation isn't completely true. In an upcoming work, \cite{DowningAshokTateo}, they show that the ansatz needs to be modified as $\kappa$ increases to accommodate for the fact that the poles no longer come in complex conjugate pairs $\bar\theta_a \neq\theta_a^*$. However, in the asymptotic regime that we're considering which allows us to expand around $\kappa=0$, they explain that what we have suffices.}
\begin{equation}
\begin{aligned}
    \epsilon(\theta)
    ={}&e^\theta+\kappa e^{5\theta}
    +\sum_{a=1}^{N}\log\!\left(
    \frac{S(\theta-\vartheta_a)}{S(\theta-\bar\vartheta_a)}
    \right)-\int_\R\varphi(\theta-\theta')
    \log\bigl(1+e^{-\epsilon(\theta')}\bigr)
    \frac{d\theta'}{2\pi}.
\end{aligned}
\end{equation}
Given $\ell_a\in\N$, the rapidities satisfy
\begin{equation}
\begin{aligned}
    2\ell_a\pi i
    =e^{\vartheta_a}+\kappa e^{5\vartheta_a}
    -\log S(\vartheta_a-\bar\vartheta_a)&+\sum_{\substack{b=1\\b\neq a}}^{N}
    \log\!\left(
    \frac{S(\vartheta_a-\vartheta_b)}
    {S(\vartheta_a-\bar\vartheta_b)}\right)\\&-\int_\R\varphi(\vartheta_a-\theta')
    \log\bigl(1+e^{-\epsilon(\theta')}\bigr)
    \frac{d\theta'}{2\pi},
\end{aligned}
\end{equation}
The corresponding proposed eigenvalues are, with $\kappa\geq0$,
\begin{equation}
\begin{aligned}
    \mathcal I_{2n-1}^{(5),\mathrm{TBA}}
    :=K_n\Bigg[&\frac{i}{2n-1}\sum_{a=1}^{N}
    \left(e^{(2n-1)\vartheta_a}-e^{(2n-1)\bar\vartheta_a}\right)+\int_\R e^{(2n-1)\theta}
    \log\bigl(1+e^{-\epsilon(\theta)}\bigr)
    \frac{d\theta}{2\pi}\Bigg].
\end{aligned}
\end{equation}
The three excited states considered in this work require only one pair of rapidities, which we denote by $\eta$ and $\bar\eta$. Following~\cite{Downing:2024nfb}, we expand
\begin{equation}\label{eq: TBA parameter expansion}
    \epsilon(\theta)=\sum_{m=0}^{\infty}\epsilon_m(\theta)\kappa^m,
    \qquad
    \eta=\sum_{m=0}^{\infty}\eta_m\kappa^m.
\end{equation}
The resulting charge eigenvalues have expansions
\begin{equation}
    \mathcal I_s^{(5),\mathrm{TBA}}\{i\}
    =\SCE_s\{i\}
    +\sum_{m=1}^{3}\kappa^m\SCE_{s,s+4m}\{i\}
    +O(\kappa^4).
\end{equation}
We define the TBA ratios by
\begin{equation}
    R_s^{\mathrm{TBA},(i)}
    :=\frac{\SCE_s\{0\}}{\SCE_s\{i\}},
    \qquad
    R_{s,s+4m}^{\mathrm{TBA},(i)}
    :=\frac{\SCE_{s,s+4m}\{0\}}{\SCE_{s,s+4m}\{i\}},
    \qquad i,m=1,2,3.
\end{equation}
We would first like to test that these are the same charges up to proportionality. This would be strongly evidenced if the ratios between the eigenvalues of the ground state with the excited states matched.
\begin{equation}
    R_s^{(i)}=R_s^{\mathrm{TBA},(i)},
    \qquad
    R_{s,s+4m}^{(i)}=R_{s,s+4m}^{\mathrm{TBA},(i)},
\end{equation}
for $s=1,5,7$ and $i,m=1,2,3$.

\subsection{Ratios from the local recursion}

The three charges to be compared are
\begin{equation}
\begin{aligned}
    \dkdvI_1^{(5)}
    &=\I_1+b_1\gamma\J_{1,5}
    +b_2\gamma^2\J_{1,9}+b_3\gamma^3\J_{1,13}
    +O(\gamma^4),\\
    \dkdvI_5^{(5)}
    &=\I_5+b_1\gamma\J_{5,9}
    +b_2\gamma^2\J_{5,13}+b_3\gamma^3\J_{5,17}
    +O(\gamma^4),\\
    \dkdvI_7^{(5)}
    &=\I_7+b_1\gamma\J_{7,11}
    +b_2\gamma^2\J_{7,15}+b_3\gamma^3\J_{7,19}
    +O(\gamma^4).
\end{aligned}
\end{equation}
Their component charges are the zero-modes
\begin{equation}
    \I_s=\sfQ{J_{s+1}},
    \qquad
    \J_{s,s+4m}=\sfQ{Y_{s,s+4m+1}},
\end{equation}
where $Y_{s,s+4m+1}$ is a quasi-primary representative of $M_{s,m}$. The fields and their low-level eigenvalues are given in Appendix~\ref{sec: Appendix with Computation}. Taking the ratios of those eigenvalues gives Table~\ref{tab:lee-yang-exact-ratios}. The repeated ratios are consistent with the proportionalities between component fields established in Proposition~\ref{prop: relation M fields}.

\begin{table}[h]
\centering
\normalsize
\renewcommand{\arraystretch}{1.3}
\begin{tabular*}{\textwidth}{@{\extracolsep{\fill}}c|ccc@{}}
\hline
Ratio & $\{1\}$ & $\{2\}$ & $\{3\}$ \\
\hline
$R_{1}^{(i)}$
& $-\tfrac{1}{59}$
& $-\tfrac{1}{119}$
& $-\tfrac{1}{179}$ \\

$R_{1,5}^{(i)}$
& $-\tfrac{89}{59131}$
& $-\tfrac{89}{1630351}$
& $-\tfrac{89}{12273571}$ \\

$R_{1,9}^{(i)}$
& $-\tfrac{1}{299}$
& $-\tfrac{1}{144599}$
& $-\tfrac{1}{5472899}$ \\

$R_{1,13}^{(i)}$
& $-\tfrac{1077983}{26025277}$
& $-\tfrac{1077983}{224811626137}$
& $-\tfrac{1077983}{42908986708597}$ \\
\hline
$R_{5}^{(i)}$
& $-\tfrac{89}{59131}$
& $-\tfrac{89}{1630351}$
& $-\tfrac{89}{12273571}$ \\

$R_{5,9}^{(i)}$
& $-\tfrac{1}{299}$
& $-\tfrac{1}{144599}$
& $-\tfrac{1}{5472899}$ \\

$R_{5,13}^{(i)}$
& $-\tfrac{1077983}{26025277}$
& $-\tfrac{1077983}{224811626137}$
& $-\tfrac{1077983}{42908986708597}$ \\

$R_{5,17}^{(i)}$
& $\tfrac{270993486853}{659028256633}$
& $-\tfrac{270993486853}{24735006743930387}$
& $-\tfrac{270993486853}{23766205518652130207}$ \\
\hline
$R_{7}^{(i)}$
& $\tfrac{211}{122371}$
& $\tfrac{211}{14154931}$
& $\tfrac{211}{238657891}$ \\

$R_{7,11}^{(i)}$
& $\tfrac{66002197}{8852004877}$
& $\tfrac{66002197}{13653213561157}$
& $\tfrac{66002197}{1160345004231037}$ \\

$R_{7,15}^{(i)}$
& $\tfrac{8045857}{81640897}$
& $\tfrac{8045857}{1266270627937}$
& $\tfrac{8045857}{541402704446977}$ \\

$R_{7,19}^{(i)}$
& $\tfrac{17117596141801}{51366264181201}$
& $\tfrac{17117596141801}{724785304506612601}$
& $\tfrac{17117596141801}{1560590830901671476001}$ \\
\hline
\end{tabular*}
\caption{Exact ratios of ground-state to excited-state coefficients for the spin-$5$ deformation. Columns $\{i\}$ refer to $L_{-i}\ket{-1/5}$, $i=1,2,3$.}
\label{tab:lee-yang-exact-ratios}
\end{table}

\subsection{Ratios from the TBA}

For the excited states under consideration, the TBA equations reduce to
\begin{equation}
\begin{aligned}
    \epsilon(\theta)
    ={}&e^\theta+\kappa e^{5\theta}
    +\log\!\left(\frac{S(\theta-\eta)}{S(\theta-\bar\eta)}\right)-\int_\R\varphi(\theta-\theta')
    \log\bigl(1+e^{-\epsilon(\theta')}\bigr)
    \frac{d\theta'}{2\pi},\\
    2\ell\pi i
    ={}&e^\eta+\kappa e^{5\eta}
    -\log S(2i \Im(\eta))-\int_\R\varphi(\eta-\theta')
    \log\bigl(1+e^{-\epsilon(\theta')}\bigr)
    \frac{d\theta'}{2\pi},
    \qquad \ell\in\N.
\end{aligned}
\end{equation}
We solve for $\epsilon_m(\theta)$ and $\eta_m$, \eqref{eq: TBA parameter expansion},through $m=3$ and substitute them into the one-pole TBA expression
\begin{equation}
\begin{aligned}
    \mathcal I_{2n-1}^{(5),\mathrm{TBA}}
    =K_n\Bigg[&\frac{i}{2n-1}
    \left(e^{(2n-1)\eta}-e^{(2n-1)\bar\eta}\right)+\int_\R e^{(2n-1)\theta}
    \log\bigl(1+e^{-\epsilon(\theta)}\bigr)
    \frac{d\theta}{2\pi}\Bigg].
\end{aligned}
\end{equation}
For the ground state, the terms involving $\eta$ are omitted. The equations for the expansion coefficients are collected in Appendix~\ref{sec: Appendix TBA}.

The numerical ratios are shown in Table~\ref{tab:lee-yang-tba-ratios} and agree with the exact ratios in Table~\ref{tab:lee-yang-exact-ratios}. The displayed values were first computed by the author in Mathematica. They were subsequently checked using ChatGPT, both with the author's method and with an alternative numerical procedure. The comparisons with the exact ratios gave relative differences of order $10^{-14}$.

\begin{table}[h]
\centering
\resizebox{\textwidth}{!}{%
\begin{tabular}{c|rrr}
\hline
Ratio & $\{1\}$ & $\{2\}$ & $\{3\}$ \\
\hline
$R_{1}^{\mathrm{TBA},(i)}$
& $-1.6949152542372(01)\times10^{-2}$
& $-8.403361344537(38)\times10^{-3}$
& $-5.586592178770(66)\times10^{-3}$ \\

$R_{1,5}^{\mathrm{TBA},(i)}$
& $-1.5051326715259(70)\times10^{-3}$
& $-5.45894718376(60)\times10^{-5}$
& $-7.251353334738(68)\times10^{-6}$ \\

$R_{1,9}^{\mathrm{TBA},(i)}$
& $-3.344481605351(26)\times10^{-3}$
& $-6.91567714852(81)\times10^{-6}$
& $-1.8271851901524(66)\times10^{-7}$ \\

$R_{1,13}^{\mathrm{TBA},(i)}$
& $-4.142061581131(43)\times10^{-2}$
& $-4.795050053786(38)\times10^{-6}$
& $-2.5122546177116(97)\times10^{-8}$ \\
\hline
$R_{5}^{\mathrm{TBA},(i)}$
& $-1.5051326715259(70)\times10^{-3}$
& $-5.45894718376(60)\times10^{-5}$
& $-7.251353334738(68)\times10^{-6}$ \\

$R_{5,9}^{\mathrm{TBA},(i)}$
& $-3.344481605351(26)\times10^{-3}$
& $-6.91567714852(81)\times10^{-6}$
& $-1.8271851901524(66)\times10^{-7}$ \\

$R_{5,13}^{\mathrm{TBA},(i)}$
& $-4.142061581131(43)\times10^{-2}$
& $-4.795050053786(38)\times10^{-6}$
& $-2.5122546177116(97)\times10^{-8}$ \\

$R_{5,17}^{\mathrm{TBA},(i)}$
& $4.1120162015133(22)\times10^{-1}$
& $-1.095586872720(40)\times10^{-5}$
& $-1.1402471742505(39)\times10^{-8}$ \\
\hline
$R_{7}^{\mathrm{TBA},(i)}$
& $1.72426473592(60)\times10^{-3}$
& $1.490646616362(91)\times10^{-5}$
& $8.841107206465(88)\times10^{-7}$ \\

$R_{7,11}^{\mathrm{TBA},(i)}$
& $7.45618624448(50)\times10^{-3}$
& $4.834187695399(12)\times10^{-6}$
& $5.688152813114(52)\times10^{-8}$ \\

$R_{7,15}^{\mathrm{TBA},(i)}$
& $9.855179567662(19)\times10^{-2}$
& $6.353979017193(55)\times10^{-6}$
& $1.4861131896669(68)\times10^{-8}$ \\

$R_{7,19}^{\mathrm{TBA},(i)}$
& $3.33245884524841(86)\times10^{-1}$
& $2.3617471319253(64)\times10^{-5}$
& $1.0968663792488(96)\times10^{-8}$ \\
\hline
\end{tabular}}
\caption{Numerical TBA ratios for the spin-$5$ deformation, with the same states and ordering as Table~\ref{tab:lee-yang-exact-ratios}. Brackets mark trailing digits of the numerical results, retained to show discrepancies with the exact ratios; they do not denote uncertainty estimates.}
\label{tab:lee-yang-tba-ratios}
\end{table}

\subsection{Matching the chemical potential}

The ratio comparison removes the overall normalisation at each order. To compare the eigenvalues themselves, we must also identify the relation between $\gamma$ and $\kappa$. Matching the first-order terms in the Hamiltonian gives
\begin{equation}
    \gamma=-5\frac{K_1}{K_3}\kappa
    =\frac{1729\pi^2}{2025\sqrt3}
    \left[\frac{\Gamma(1/6)}{\Gamma(2/3)}\right]^6\kappa.
\end{equation}
More generally, first-order matching for a spin-$\s$ insertion gives
\begin{equation}
    \kappa_\s
    =-\frac{K_{(\s+1)/2}}{\s K_1}\,\gamma
    =\frac{2\pi}{\s}K_{(\s+1)/2}\,\gamma,
    \qquad K_1=-\frac{1}{2\pi}.
\end{equation}
For the spin-$5$ deformation, substituting this relation into the local expansion gives
\begin{equation}
    \dkdvI_s^{(5)}
    =\I_s+\sum_{m=1}^{3}
    b_m\left(-5\frac{K_1}{K_3}\right)^m
    \kappa^m\J_{s,s+4m}
    +O(\kappa^4).
\end{equation}
The ground-state coefficients must therefore satisfy
\begin{equation}\label{eq: Need to satisfy, SCE= J}
\begin{aligned}
    \SCE_s\{0\}&=I_s\{0\},\\
    \SCE_{s,s+4m}\{0\}
    &=b_m\left(-5\frac{K_1}{K_3}\right)^m
    J_{s,s+4m}\{0\},
    \qquad m=1,2,3.
\end{aligned}
\end{equation}
Table~\ref{tab:lee-yang-tba-ground-state-coefficients} lists the coefficients obtained from the TBA. Comparing them with the exact eigenvalues in Table~\ref{tab:lee-yang-low-level-charge-eigenvalues} of Appendix~\ref{sec: Appendix with Computation} confirms~\eqref{eq: Need to satisfy, SCE= J} to high numerical accuracy.

\begin{table}[h]
\centering
\renewcommand{\arraystretch}{1.3}
\resizebox{\textwidth}{!}{%
\begin{tabular}{cr||cr||cr}
\hline
Coefficient & Numerical value
& Coefficient & Numerical value
& Coefficient & Numerical value \\
\hline
$\SCE_{1}\{0\}$ & $-1.666666666666(58)\times10^{-2}$
& $\SCE_{5}\{0\}$ & $1.177248677248(70)\times10^{-4}$
& $\SCE_{7}\{0\}$ & $-2.6049382716049(99)\times10^{-5}$ \\

$\SCE_{1,5}\{0\}$ & $2.76341105920(90)\times10^{0}$
& $\SCE_{5,9}\{0\}$ & $-1.8789642724498(45)\times10^{0}$
& $\SCE_{7,11}\{0\}$ & $1.4861823744250(51)\times10^{0}$ \\

$\SCE_{1,9}\{0\}$ & $-4.410581001867(35)\times10^{4}$
& $\SCE_{5,13}\{0\}$ & $2.29156127826(10)\times10^{5}$
& $\SCE_{7,15}\{0\}$ & $-3.935960834946(96)\times10^{5}$ \\

$\SCE_{1,13}\{0\}$ & $5.37908931356(90)\times10^{9}$
& $\SCE_{5,17}\{0\}$ & $-1.0732110885526(49)\times10^{11}$
& $\SCE_{7,19}\{0\}$ & $3.221077635354(54)\times10^{11}$ \\
\hline
\end{tabular}}
\caption{Ground-state coefficients in the TBA expansion through cubic order in $\kappa$ for the spin-$5$ deformation. Brackets mark numerical digits where they begin to diverge from the exact answer, as in Table~\ref{tab:lee-yang-tba-ratios}.}
\label{tab:lee-yang-tba-ground-state-coefficients}
\end{table}

Together with the ratio comparison, this establishes numerical agreement between the local recursion and the proposed TBA eigenvalues for $s=1,5,7$, through cubic order in the spin-$5$ deformation, on all four states considered. These results support the TBA interpretation of the higher deformed charges in a setting where the component charges within a given deformed charge need not commute individually.

\section{Conclusion and Outlook}

Let us review the content of this work. We have constructed the $\sigma$-deformed KdV hierarchy, an infinite family of mutually commuting charges motivated by the modular transformation of a GGE with a spin-$\sigma$ KdV charge inserted. The charges are defined as formal power series in the chemical potential, with coefficients given by local integrals of motion. The first member reproduces the deformed Hamiltonian appearing in the asymptotic modular transform, while the higher members extend this construction to the remaining KdV charges.

The hierarchy is specified by local recursion relations, which can also be expressed as flow equations. These provide two complementary proofs of commutativity to all orders and establish the existence of the local primitives required by the construction. For a fixed generating series, the deformed currents are uniquely determined modulo total derivatives. Although the individual charges appearing as coefficients need not commute, their contributions combine into a commuting family of formal series.

We tested the TBA interpretation of these charges in the free fermion and Lee-Yang models. In the free fermion theory, the hierarchy reduces to explicit linear combinations of undeformed KdV charges, and its ground-state eigenvalues in both the Ramond and Neveu-Schwarz sectors agree with the TBA expansion to all orders. In the Lee-Yang model, we found numerical agreement through cubic order in the spin-$5$ deformation for the spin-$1$, spin-$5$, and spin-$7$ charges on the ground state and three low-lying excited states. These results support the interpretation of the higher deformed charges through the TBA in the setting of an non-trivial interacting theory.

Several extensions suggest themselves. Although we have focused on the KdV hierarchy, the construction is formulated in terms of OPEs and the local commutativity of the undeformed seed currents. It is therefore natural to investigate whether the same approach applies to other integrable hierarchies, such as the Boussinesq and Bullough--Dodd hierarchies. A particularly accessible example is the $W_{1+\infty}$ hierarchy at $c=-2$ discussed in~\cite{Karimi:2026tvh}. Its similar pre-Lie structure to the free fermion case shows that the corresponding recursion will also admit an explicit solution.

Furthermore, there is no reason to restrict ourselves to one charge. GGEs with several inserted charges may lead to multi-parameter versions of the flows in Proposition \ref{prop: ao flows}. 

For the free fermion, one can ask whether the deformed CFT charges extend to the tensor currents of the massive theory studied in \cite{Downing:2023uuc}. The massive currents in that paper reduce to the KdV charges in the massless limit.

It would be interesting to conduct an investigation of the deformed charges through non-linear integral equations of Destri-de Vega type \cite{Bazhanov:1996dr} and through the ODE/IM correspondence. The former direction connects with existing literature on NLIE GGEs in the sine-gordon model \cite{Hegedus:2025yhr}. The latter direction connects with existing interest in GGEs and ODE/IM~\cite{Masoero:2010is,Brizio:2025fms}. 

For the ODE/IM correspondence, the free fermion model would be very accessible, since in that case the deformation reorganises the undeformed KdV charges into a series in the chemical potential. This suggests investigating whether the same mixing can be realised on the schrodinger equation side by a coupling-dependent change of spectral parameter, $E\mapsto f(E,\kappa_\s)$. This would work to reshuffle the charges appearing in the asymptotic expansion of the $Q$-function. One expects a similar story for the symplectic fermion model, although the function $f$ would be different since at $c=-2$ the potential is the $x^2$ harmonic oscillator, whereas at $c=1/2$ the potential is $x^{2/3}$. Determining such functions and matching the resulting expansion would provide a concrete first test. 

In a companion paper, released alongside the present work, we investigate the Lee-Yang GGE in the context of the ODE/IM correspondence~\cite{Ashok:2026mfy}, building in part on the results established here. This case is not as simple as the above proposition for the free fermion, since the deformed charges are not made out of sums of undeformed KdV charges. One needs to change the spectral data by adding terms linear in the ODE coordinate to the potential. They found that the period integrals of the deformed ODE give exactly the deformed KdV charge eigenvalues for both the $h=-1/5$ and $h=0$ primary states. A further direction to pursue in connection with this, and the present work, is to construct deformed hierarchies associated to more general integrable hierarchies, like those associated to $W$-algebras and calculate periods of the higher rank ODEs like in \cite{Kudrna:2025bzg}.

These approaches also motivate a functional formulation of the GGE, potentially through a modification of the BLZ construction~\cite{Bazhanov:1994ft,Bazhanov:1996dr}. Indeed it is already known that at least the Y-functions of GGEs satisfy the same functional relations as in the undeformed case \cite{Castro-Alvaredo:2022pgz}. The aim would be to identify functional relations whose asymptotic expansions generate the deformed charges. Such a formulation could connect the local recursion to the spectral descriptions supplied by the TBA, non-linear integral equations, and ODE/IM.

An unresolved question is the non-perturbative completion of the hierarchy. The present construction determines formal expansions around vanishing chemical potential, whereas the full modular transform can contain additional contributions that are invisible to these expansions. The non-perturbative TBA solutions discussed in~\cite{Downing:2024nfb} offer a starting point, and there is ongoing work looking to understand the full structure of these non-perturbative solutions in more detail \cite{DowningAshokTateo}. Based on the analysis in these works, and the free cases where they were worked out in detail \cite{Downing:2021mfw,Downing:2023lop} in general we expect that the spectrum of the hierarchy actually becomes complex. Understanding how they contribute to the higher charges, and whether the resulting completed charges retain the commutativity established perturbatively here, would clarify the relation between the formal hierarchy and the full modular-transformed GGE. 

We said before that the GGE and its modular transformation can be interpreted as a translation invariant defect. A further question is whether this defect and its higher-spin charges coincide with those studied in \cite{Ambrosino:2025myh}.

Finally, a microscopic description in a statistical model could give a more direct interpretation of both the modular transformation and the deformed charges. It may also clarify the relation between the OPE flow equations and the boost deformations studied in spin chains~\cite{Bargheer:2008jt}. A lattice formulation would provide a useful setting in which to develop the proposed functional approach, with the free theories at $c=1/2$ and $c=-2$ as natural starting points. The latter is particularly relevant in view of previous work on $c=-2$ GGEs and integrals of motion in statistical physics~\cite{Baseilhac:2025upc, Nigro:2009si}.

\section*{Acknowledgements}
FK thanks M. Downing, G. M. T. Watts, S. Ashok, and R. Tateo for helpful discussions and for reading early drafts of this paper. FK is also grateful to S. Brian and S. Bartlett-Tisdall for their useful scrutiny of the manuscript.
FK would like to acknowledge the funding from the STFC grant ST/Y509279/1 in supporting this work.

\section*{Statement on use of Artificial Intelligence}
The results that were partially obtained by AI, and the extent to which AI was used to obtain them have been clearly marked as such by the author. Conversation with AI took place on OpenAI's ChatGPT 6 Astra. 

\appendix

\section{Fields, Charges, and Eigenvalues in the Lee-Yang Model}\label{sec: Appendix with Computation}

This appendix lists the local fields and charge eigenvalues used in Section~\ref{sec: The Deformed Hierarchy in LY}. All expressions are evaluated in the irreducible Lee-Yang modules, with the null relations imposed. A basis of the irreducible vacuum module is given in~\cite{Bajnok:2014fca} by
\begin{equation}
    \CH_0=\operatorname{span}\bigl\{
    L_{-n_1}\cdots L_{-n_m}\ket{0}
    \;\big|\;
    m\geq0,\quad n_m\geq2,\quad n_i-n_{i+1}\geq2
    \bigr\}.
\end{equation}
The empty product at $m=0$ denotes the vacuum; the restrictions on the indices apply when they are present. We choose quasi-primary representatives of the KdV currents. After imposing the null relations, the currents needed here take the following form in BLZ normalisation:
\begin{equation}
\begin{aligned}
    J_2&=T,\\
    J_4&=0,\\
    J_6&=\frac1{10}
    \left[(\partial^2T\,T)-\frac1{42}\partial^4T\right],\\
    J_8&=-\frac7{4875}
    \left[(\partial^4T\,T)
    +14(\partial^3T\,\partial T)
    -\frac7{90}\partial^6T\right].
\end{aligned}
\end{equation}
For the spin-$5$ deformation, the generating field is $W=J_6$, and the seed fields for the charges studied in the main text are $J_2$, $J_6$, and $J_8$. To express the fields compactly, associate a field $\mathcal B_{a_1,\ldots,a_k}$ with each basis state
\begin{equation}
    \ket{\mathcal B_{a_1,\ldots,a_k}}
    :=L_{-a_1}\cdots L_{-a_k}\ket{0},
    \qquad a_k\geq2,\quad a_i-a_{i+1}\geq2.
\end{equation}
The corresponding fields are defined recursively by
\begin{equation}
\begin{aligned}
    \mathcal B_a
    &=\frac{\partial^{a-2}T}{(a-2)!},\\
    \mathcal B_{a_1,\ldots,a_k}
    &=\frac1{(a_1-2)!}
    \bigl(\partial^{a_1-2}T\,
    \mathcal B_{a_2,\ldots,a_k}\bigr),
    \qquad k\geq2.
\end{aligned}
\end{equation}
Products in brackets denote conformal normal ordering. Applying the local recursion in this basis gives the quasi-primary correction fields listed in~\eqref{tab: Currents from Recursion}. Their notation is related to the fields of Section~\ref{sec: The Deformed Hierarchy} by
\begin{equation}
    Y_{s,s+4m+1}= M_{s,m}\pmod{\del},
    \qquad
    \J_{s,s+4m}=\sfQ{Y_{s,s+4m+1}}.
\end{equation}

\begin{equation}\label{tab: Currents from Recursion}
\begin{array}{c|l}
\text{Field} & \text{Exact expression} \\ \hline
Y_{1,6}  & J_6 \\[6pt]
Y_{1,10} & -\dfrac{546}{116875}\mathcal P_{10} \\[8pt]
Y_{1,14} & \dfrac{13104}{6828125}\mathcal P_{14} \\[8pt]
Y_{5,10} & -\dfrac{273}{116875}\mathcal P_{10} \\[8pt]
Y_{5,14} & \dfrac{3276}{6828125}\mathcal P_{14} \\[8pt]
Y_{5,18} & -\dfrac{10192}{543740484375}\mathcal P_{18} \\[8pt]
Y_{7,12} & \dfrac{28}{39365625}\mathcal P_{12} \\[8pt]
Y_{7,16} & -\dfrac{5422144}{2657578125}\mathcal P_{16} \\[8pt]
Y_{7,20} & \dfrac{81536}{60743955078125}\mathcal P_{20}
\end{array}
\end{equation}
The fields $\mathcal P_h$ appearing in these expressions are given below, together with the weight-$6$ field. These expressions were obtained with the assistance of ChatGPT; the first two cases were independently checked directly, while the rest were checked by inspecting and testing the code produced. The cylinder charges and their low-level eigenvalues were computed using the method of Appendix B of~\cite{Downing:2024nfb}. The results are tested against the TBA in the main body of this paper.

\medskip
\noindent\textbf{Weight 6.}
\begin{equation}
\begin{aligned}
\mathcal P_{6}&=-2\mathcal B_{6} + 7\mathcal B_{4,2}.
\end{aligned}
\end{equation}
\medskip
\noindent\textbf{Weight 10.}
\begin{equation}
\begin{aligned}
\mathcal P_{10}&=-756\mathcal B_{10} + 440\mathcal B_{8,2}\\
&\quad +165\mathcal B_{7,3} + 660\mathcal B_{6,4}.
\end{aligned}
\end{equation}
\medskip
\noindent\textbf{Weight 12.}
\begin{equation}
\begin{aligned}
\mathcal P_{12}&=-19845952\mathcal B_{12} + 4592276\mathcal B_{10,2}\\
&\quad +7996495\mathcal B_{9,3} + 3140826\mathcal B_{8,4}\\
&\quad +10136269\mathcal B_{7,5} + 2351440\mathcal B_{6,4,2}.
\end{aligned}
\end{equation}
\medskip
\noindent\textbf{Weight 14.}
\begin{equation}
\begin{aligned}
\mathcal P_{14}&=-2658804\mathcal B_{14} + 483885\mathcal B_{12,2}\\
&\quad +548334\mathcal B_{11,3} + 777213\mathcal B_{10,4}\\
&\quad +246039\mathcal B_{9,5} + 1040537\mathcal B_{8,6}\\
&\quad +30590\mathcal B_{8,4,2} + 137655\mathcal B_{7,5,2}.
\end{aligned}
\end{equation}
\medskip
\noindent\textbf{Weight 16.}
\begin{equation}
\begin{aligned}
\mathcal P_{16}&=-8673291\mathcal B_{16} + 919564\mathcal B_{14,2}\\
&\quad +1630555\mathcal B_{13,3} + 1239487\mathcal B_{12,4}\\
&\quad +1958434\mathcal B_{11,5} + 474946\mathcal B_{10,6}\\
&\quad +228735\mathcal B_{10,4,2} + 2612203\mathcal B_{9,7}\\
&\quad +64515\mathcal B_{9,5,2} + 322575\mathcal B_{8,6,2}.
\end{aligned}
\end{equation}
\medskip
\noindent\textbf{Weight 18.}
\begin{equation}
\begin{aligned}
\mathcal P_{18}&=-367359790991556\mathcal B_{18} + 25827525172432\mathcal B_{16,2}\\
&\quad +48495860630182\mathcal B_{15,3} + 54286890328090\mathcal B_{14,4}\\
&\quad +38435988850837\mathcal B_{13,5} + 65820212896372\mathcal B_{12,6}\\
&\quad +3401688198240\mathcal B_{12,4,2} + 14983388827297\mathcal B_{11,7}\\
&\quad +6908557697055\mathcal B_{11,5,2} + 87804128523282\mathcal B_{10,8}\\
&\quad +2633747844780\mathcal B_{10,6,2} + 915342605685\mathcal B_{10,5,3}\\
&\quad +9412798400280\mathcal B_{9,7,2} + 538025806395\mathcal B_{9,6,3}\\
&\quad -768608294850\mathcal B_{8,6,4}.
\end{aligned}
\end{equation}
\medskip
\noindent\textbf{Weight 20.}
\begin{equation}
\begin{aligned}
\mathcal P_{20}&=-130361354017883868\mathcal B_{20} + 5331603449881520\mathcal B_{18,2}\\
&\quad +14085717813303045\mathcal B_{17,3} + 14154260224180530\mathcal B_{16,4}\\
&\quad +15073436263187375\mathcal B_{15,5} + 10093296304465480\mathcal B_{14,6}\\
&\quad +1211222775472500\mathcal B_{14,4,2} + 19342450858824570\mathcal B_{13,7}\\
&\quad +1115643946810500\mathcal B_{13,5,2} + 3656926491717620\mathcal B_{12,8}\\
&\quad +2140792284414500\mathcal B_{12,6,2} + 348661539735000\mathcal B_{12,5,3}\\
&\quad +25558019854881630\mathcal B_{11,9} + 561523016422250\mathcal B_{11,7,2}\\
&\quad +125939385557000\mathcal B_{11,6,3} + 2528671673796500\mathcal B_{10,8,2}\\
&\quad +132060511199000\mathcal B_{10,7,3} - 52575623428000\mathcal B_{10,6,4}\\
&\quad -254986034259000\mathcal B_{9,7,4} - 34179894595000\mathcal B_{8,6,4,2}.
\end{aligned}
\end{equation}

The cylinder zero-modes of these fields act on the $h=-1/5$ module. Table~\ref{tab:lee-yang-low-level-charge-eigenvalues} gives their eigenvalues on
\begin{equation}
    \ket{h},\qquad L_{-1}\ket{h},\qquad
    L_{-2}\ket{h},\qquad L_{-3}\ket{h},
    \qquad h=-\frac15.
\end{equation}
These states are labelled by $\{0\},\{1\},\{2\},\{3\}$, respectively, as in the main text. The ground-state to excited-state ratios used in Section~\ref{sec: The Deformed Hierarchy in LY} follow directly from this table.

\begin{table}[h]
\centering
\renewcommand{\arraystretch}{1.9}
\resizebox{\textwidth}{!}{$
\begin{array}{c|rrrr}
\text{Charge} & \{0\} & \{1\} & \{2\} & \{3\} \\ \hline
\I_{1} & -\frac{1}{60} & \frac{59}{60} & \frac{119}{60} & \frac{179}{60} \\
\I_{3} & 0 & 0 & 0 & 0 \\
\I_{5} & \frac{89}{756000} & -\frac{59131}{756000} & -\frac{1630351}{756000} & -\frac{12273571}{756000} \\
\I_{7} & -\frac{211}{8100000} & -\frac{122371}{8100000} & -\frac{14154931}{8100000} & -\frac{238657891}{8100000} \\
\J_{1,5} & \frac{89}{756000} & -\frac{59131}{756000} & -\frac{1630351}{756000} & -\frac{12273571}{756000} \\
\J_{1,9} & -\frac{1729}{10800000} & \frac{516971}{10800000} & \frac{250011671}{10800000} & \frac{9462642371}{10800000} \\
\J_{1,13} & \frac{1077983}{324000000} & -\frac{26025277}{324000000} & -\frac{224811626137}{324000000} & -\frac{42908986708597}{324000000} \\
\J_{5,9} & -\frac{1729}{21600000} & \frac{516971}{21600000} & \frac{250011671}{21600000} & \frac{9462642371}{21600000} \\
\J_{5,13} & \frac{1077983}{1296000000} & -\frac{26025277}{1296000000} & -\frac{224811626137}{1296000000} & -\frac{42908986708597}{1296000000} \\
\J_{5,17} & -\frac{270993486853}{4082400000000} & -\frac{659028256633}{4082400000000} & \frac{24735006743930387}{4082400000000} & \frac{23766205518652130207}{4082400000000} \\
\J_{7,11} & \frac{66002197}{1042470000000} & \frac{8852004877}{1042470000000} & \frac{13653213561157}{1042470000000} & \frac{1160345004231037}{1042470000000} \\
\J_{7,15} & -\frac{97379007271}{68161500000000} & -\frac{988099776391}{68161500000000} & -\frac{15325673409921511}{68161500000000} & -\frac{6552596931921762631}{68161500000000} \\
\J_{7,19} & \frac{119823172992607}{601425000000000} & \frac{359563849268407}{601425000000000} & \frac{5073497131546288207}{601425000000000} & \frac{10924135816311700332007}{601425000000000} \\
\end{array}
$}\caption{Eigenvalues on $\ket{h}$, $L_{-1}\ket{h}$, $L_{-2}\ket{h}$, and $L_{-3}\ket{h}$, labelled by $\{0\},\{1\},\{2\},\{3\}$, respectively, with $h=-1/5$ and $c=-22/5$. The charges are the component piece of the deformed KdV charges. The cylinder charges are defined by $\sfQ{Y}=\int_0^{2\pi i}Y(w)\,dw/(2\pi i)$.}
\label{tab:lee-yang-low-level-charge-eigenvalues}
\end{table}

\section{TBA Expansion of Lee-Yang Charge Eigenvalues}\label{sec: Appendix TBA}

We collect the equations used to compute the small-$\kappa$ expansion of the TBA eigenvalues through cubic order. They follow the expansion method of~\cite{Downing:2024nfb}, with the coupling normalisation specified below. We consider the spin-$5$ deformation and the excited states described by a single pair of complex rapidities $\eta,\bar\eta$. Write
\begin{equation}
    \epsilon(\theta)=\sum_{m=0}^{3}\epsilon_m(\theta)\kappa^m
    +O(\kappa^4),
    \qquad
    \eta=\sum_{m=0}^{3}\eta_m\kappa^m+O(\kappa^4),
\end{equation}
and define
\begin{equation}
    L(x):=\log(1+e^{-x}).
\end{equation}
Primes on $L$ and $\varphi$ denote derivatives with respect to their arguments. In the integrands of the pseudoenergy and rapidity equations below, we suppress the argument $\theta'$ of $\epsilon_m$; thus $L^{(r)}(\epsilon_0)$ means $L^{(r)}(\epsilon_0(\theta'))$. The functions $S$ and $\varphi$ are those defined in Section~\ref{sec: The Deformed Hierarchy in LY}.

Expanding the pseudoenergy equation gives
\begin{equation}\label{eq:epsilon0excited}
\begin{aligned}
    \epsilon_0(\theta)
    ={}&e^\theta
    +\log\!\left(\frac{S(\theta-\eta_0)}
    {S(\theta-\bar\eta_0)}\right)-\int_\R\varphi(\theta-\theta')
    L(\epsilon_0)\frac{d\theta'}{2\pi},
\end{aligned}
\end{equation}
\begin{equation}\label{eq: first order epsilon}
\begin{aligned}
    \epsilon_1(\theta)
    ={}&e^{5\theta}
    +2\Im\!\left(\eta_1\varphi(\theta-\eta_0)\right)\\
    &-\int_\R\varphi(\theta-\theta')
    \epsilon_1L'(\epsilon_0)\frac{d\theta'}{2\pi},
\end{aligned}
\end{equation}
\begin{equation}
\begin{aligned}
    \epsilon_2(\theta)
    ={}&2\Im\!\left(
    \eta_2\varphi(\theta-\eta_0)
    -\frac12\eta_1^2\varphi'(\theta-\eta_0)\right)-\int_\R\varphi(\theta-\theta')
    \left[\epsilon_2L'(\epsilon_0)
    +\frac12\epsilon_1^2L''(\epsilon_0)\right]
    \frac{d\theta'}{2\pi},
\end{aligned}
\end{equation}
and
\begin{equation}
\begin{aligned}
    \epsilon_3(\theta)
    ={}&2\Im\Bigl(
    \eta_3\varphi(\theta-\eta_0)
    -\eta_2\eta_1\varphi'(\theta-\eta_0)
    +\frac16\eta_1^3\varphi''(\theta-\eta_0)\Bigr)\\
    &-\int_\R\varphi(\theta-\theta')
    \Bigl[\epsilon_3L'(\epsilon_0)
    +\epsilon_2\epsilon_1L''(\epsilon_0)
    +\frac16\epsilon_1^3L'''(\epsilon_0)\Bigr]
    \frac{d\theta'}{2\pi}.
\end{aligned}
\end{equation}
The ground-state equations are obtained by omitting all source terms involving the coefficients $\eta_m$ of the pole and their conjugates. For an excited state with fixed $\ell\in\N$, the zeroth-order rapidity satisfies
\begin{equation}\label{eq:eta0constraint}
\begin{aligned}
    2\ell\pi i
    ={}&e^{\eta_0}-\log S\!\left(2i\Im(\eta_0)\right)-\int_\R\varphi(\eta_0-\theta')
    L(\epsilon_0)\frac{d\theta'}{2\pi}.
\end{aligned}
\end{equation}
The first-order equation is
\begin{equation}\label{eq: first order}
\begin{aligned}
    0={}&\eta_1e^{\eta_0}+e^{5\eta_0}
    +2\Im(\eta_1)\varphi\!\left(2i\Im(\eta_0)\right)\\
    &-\int_\R\Bigl[
    \eta_1\varphi'(\eta_0-\theta')L(\epsilon_0)
    +\varphi(\eta_0-\theta')\epsilon_1L'(\epsilon_0)
    \Bigr]\frac{d\theta'}{2\pi}.
\end{aligned}
\end{equation}
At second order,
\begin{equation}
\begin{aligned}
    0={}&\left(\eta_2+\frac12\eta_1^2\right)e^{\eta_0}
    +5\eta_1e^{5\eta_0}+2\Bigl[
    \Im(\eta_2)\varphi\!\left(2i\Im(\eta_0)\right)
    +i\Im(\eta_1)^2\varphi'\!\left(2i\Im(\eta_0)\right)
    \Bigr]\\
    &-\int_\R\Biggl\{
    \left[\eta_2\varphi'(\eta_0-\theta')
    +\frac12\eta_1^2\varphi''(\eta_0-\theta')\right]
    L(\epsilon_0)
    +\eta_1\varphi'(\eta_0-\theta')\epsilon_1L'(\epsilon_0)\\
    &\hspace{17mm}
    +\varphi(\eta_0-\theta')
    \left[\epsilon_2L'(\epsilon_0)
    +\frac12\epsilon_1^2L''(\epsilon_0)\right]
    \Biggr\}\frac{d\theta'}{2\pi}.
\end{aligned}
\end{equation}
At third order,
\begin{equation}
\begin{aligned}
    0={}&\left(\eta_3+\eta_2\eta_1+\frac16\eta_1^3\right)e^{\eta_0}
    +\left(5\eta_2+\frac{25}{2}\eta_1^2\right)e^{5\eta_0}\\
    &+2\Im(\eta_3)\varphi\!\left(2i\Im(\eta_0)\right)
    +4i\Im(\eta_2)\Im(\eta_1)
    \varphi'\!\left(2i\Im(\eta_0)\right)-\frac43\Im(\eta_1)^3
    \varphi''\!\left(2i\Im(\eta_0)\right)\\
    &-\int_\R\Biggl\{
    \left[\eta_3\varphi'(\eta_0-\theta')
    +\eta_2\eta_1\varphi''(\eta_0-\theta')
    +\frac16\eta_1^3\varphi'''(\eta_0-\theta')\right]
    L(\epsilon_0)\\
    &\hspace{17mm}
    +\left[\eta_2\varphi'(\eta_0-\theta')
    +\frac12\eta_1^2\varphi''(\eta_0-\theta')\right]
    \epsilon_1L'(\epsilon_0)\\
    &\hspace{17mm}
    +\eta_1\varphi'(\eta_0-\theta')
    \left[\epsilon_2L'(\epsilon_0)
    +\frac12\epsilon_1^2L''(\epsilon_0)\right]\\
    &\hspace{17mm}
    +\varphi(\eta_0-\theta')
    \Bigl[\epsilon_3L'(\epsilon_0)
    +\epsilon_1\epsilon_2L''(\epsilon_0)
    +\frac16\epsilon_1^3L'''(\epsilon_0)\Bigr]
    \Biggr\}\frac{d\theta'}{2\pi}.
\end{aligned}
\end{equation}
These equations are required only for the excited states. 

Our convention absorbs the constant multiplying the spin-$5$ driving term into $\kappa$. To compare with the expansion parameter $\lambda:=\alpha/L^4$ used in~\cite{Downing:2024nfb}, set
\begin{equation}
    \kappa=C\lambda,
\end{equation}
where $C$ is the constant defined there and $L$ in $\alpha/L^4$ denotes the length scale, not the function $L(x)$ introduced above. The driving term can then be written as
\begin{equation}
    e^\theta+\kappa e^{5\theta}
    =e^\theta+C\lambda e^{5\theta}.
\end{equation}
If $\widetilde\epsilon_m$ and $\widetilde\eta_m$ denote coefficients of powers of $\lambda$, the conversion is
\begin{equation}
    \widetilde\epsilon_m=C^m\epsilon_m,
    \qquad
    \widetilde\eta_m=C^m\eta_m.
\end{equation}
In particular, the first-order driving terms in~\eqref{eq: first order epsilon} and~\eqref{eq: first order} become $Ce^{5\theta}$ and $Ce^{5\eta_0}$. At higher orders, the terms obtained by expanding $C\lambda e^{5\eta}$ carry the same explicit factor $C$. The coefficient rescaling above accounts for the conversion at every order.

The TBA eigenvalue associated with a one-pole solution is
\begin{equation}
\begin{aligned}
    \mathcal I_{2n-1}^{(5),\mathrm{TBA}}
    =K_n\Bigg[&\frac{i}{2n-1}
    \left(e^{(2n-1)\eta}-e^{(2n-1)\bar\eta}\right)+\int_\R e^{(2n-1)\theta}
    L(\epsilon(\theta))\frac{d\theta}{2\pi}\Bigg],
    \qquad \kappa\geq0.
\end{aligned}
\end{equation}
We write its expansion as
\begin{equation}
    \mathcal I_s^{(5),\mathrm{TBA}}
    =\SCE_s+\kappa\SCE_{s,s+4}
    +\kappa^2\SCE_{s,s+8}
    +\kappa^3\SCE_{s,s+12}
    +O(\kappa^4).
\end{equation}
The state label is suppressed here; each choice of the TBA solution gives the corresponding coefficients. Substituting the expansions of $\epsilon$ and $\eta$ yields the following expressions. In these integrals, every $\epsilon_m$ is evaluated at $\theta$.

\begin{equation}
    \SCE_{2n-1}
    =K_n\left[
    -\frac{2}{2n-1}\Im\!\left(e^{(2n-1)\eta_0}\right)
    +\int_\R e^{(2n-1)\theta}
    L(\epsilon_0)\frac{d\theta}{2\pi}\right],
\end{equation}
\begin{equation}
    \SCE_{2n-1,\,2n+3}
    =K_n\left[
    -2\Im\!\left(e^{(2n-1)\eta_0}\eta_1\right)
    +\int_\R e^{(2n-1)\theta}
    \epsilon_1L'(\epsilon_0)\frac{d\theta}{2\pi}\right],
\end{equation}
\begin{equation}
\begin{aligned}
    \SCE_{2n-1,\,2n+7}
    =K_n\Bigg[&-2\Im\!\left(
    e^{(2n-1)\eta_0}
    \left[\eta_2+\frac{2n-1}{2}\eta_1^2\right]\right)\\
    &+\int_\R e^{(2n-1)\theta}
    \left[\epsilon_2L'(\epsilon_0)
    +\frac12\epsilon_1^2L''(\epsilon_0)\right]
    \frac{d\theta}{2\pi}\Bigg],
\end{aligned}
\end{equation}
and
\begin{equation}
\begin{aligned}
    \SCE_{2n-1,\,2n+11}
    =K_n\Bigg[&-2\Im\!\left(
    e^{(2n-1)\eta_0}
    \left[\eta_3+(2n-1)\eta_1\eta_2
    +\frac{(2n-1)^2}{6}\eta_1^3\right]\right)\\
    &+\int_\R e^{(2n-1)\theta}
    \Bigl[\epsilon_3L'(\epsilon_0)
    +\epsilon_1\epsilon_2L''(\epsilon_0)
    +\frac16\epsilon_1^3L'''(\epsilon_0)\Bigr]
    \frac{d\theta}{2\pi}\Bigg].
\end{aligned}
\end{equation}
The ground-state coefficients follow by omitting the any term containing $\eta_i$. Together with the solutions of the preceding equations, these formulae give the coefficients and ratios compared in Section~\ref{sec: The Deformed Hierarchy in LY}.

\newpage
\bibliography{bib}

\end{document}